\documentclass[a4paper,USenglish]{lipics-v2019}
\nolinenumbers

\usepackage{microtype}
\usepackage{dsfont}
\usepackage{mathtools}
\usepackage{amsmath,amssymb,amsthm}

\newtheorem{hypothesis}[theorem]{Induction Hypothesis}
\newtheorem{basecase}[theorem]{Base Case}

\newcommand\soaalloc{\textsc{DynaSOAr}}

\title{DynaSOAr: A Parallel Memory Allocator for Object-oriented Programming on GPUs with Efficient Memory Access}
\titlerunning{\soaalloc{}}

\author{Matthias Springer}{Tokyo Institute of Technology}{matthias.springer@acm.org}{}{}

\author{Hidehiko Masuhara}{Tokyo Institute of Technology}{masuhara@acm.org}{}{}

\authorrunning{M. Springer and H. Masuhara}

\Copyright{Matthias Springer and Hidehiko Masuhara}

\ccsdesc[500]{Software and its engineering~Allocation / deallocation strategies}
\ccsdesc[300]{Software and its engineering~Object oriented languages}
\ccsdesc[300]{Computer systems organization~Single instruction, multiple data}

\keywords{CUDA, Data Layout, Dynamic Memory Allocation, GPUs, Object-oriented Programming, SIMD, Single-Instruction Multiple-Objects, Structure of Arrays} 

\category{}

\relatedversion{}

\acknowledgements{This work was supported by a JSPS Research Fellowship for Young Scientists (KAKENHI Grant Number 18J14726). We gratefully acknowledge the support of NVIDIA Corporation with the donation of the TITAN Xp GPU used for this research. We would also like to thank Hoang NT, Jonathon D. Tanks and the anonymous reviewers for their comments and suggestions on earlier versions of this paper.} 

\supplement{Source code: \url{https://github.com/prg-titech/dynasoar}.}

\EventEditors{Alastair F. Donaldson}
\EventNoEds{1}
\EventLongTitle{33rd European Conference on Object-Oriented Programming (ECOOP 2019)}
\EventShortTitle{ECOOP 2019}
\EventAcronym{ECOOP}
\EventYear{2019}
\EventDate{July 15--19, 2019}
\EventLocation{London, United Kingdom}
\EventLogo{}
\SeriesVolume{134}
\ArticleNo{16}

\usepackage{tikz}
\usetikzlibrary{tikzmark}
\usetikzlibrary{positioning}
\usetikzlibrary{shapes.callouts}

\tikzset{
  level/.style   = { ultra thick, blue },
  connect/.style = { dashed, red },
  notice/.style  = { draw, rectangle callout, callout relative pointer={#1} },
  label/.style   = { text width=2cm }
}

\DeclareMathOperator*{\argmin}{argmin}
\usepackage[linesnumbered, ruled]{algorithm2e} 
\usepackage{multicol}

\usepackage{wasysym}
\usepackage{amssymb}
\usepackage{amsmath}
\usepackage{pifont}
\usepackage{subfloat}
\usepackage{float}
\usepackage{varwidth}

\newcommand\narrowstyle{\SetTracking{encoding=*}{-50}\lsstyle}

\definecolor{mygreen}{rgb}{0,0.6,0}
\definecolor{mygray}{rgb}{0.5,0.5,0.5}
\definecolor{mymauve}{rgb}{0.58,0,0.82}

\SetAlFnt{\sffamily}
\usepackage{makecell}

\begin{document}

\maketitle

\begin{abstract}
Object-oriented programming has long been regarded as too inefficient for SIMD high-performance computing, despite the fact that many important HPC applications have an inherent object structure. On SIMD accelerators, including GPUs, this is mainly due to performance problems with memory allocation and memory access: There are a few libraries that support parallel memory allocation directly on accelerator devices, but all of them suffer from uncoalesed memory accesses.




We discovered a broad class of object-oriented programs with many important real-world applications that can be implemented efficiently on massively parallel SIMD accelerators. We call this class \emph{Single-Method Multiple-Objects} (SMMO), because parallelism is expressed by running a method on all objects of a type.

To make fast GPU programming available to domain experts who are less experienced in GPU programming, we developed \textsc{DynaSOAr}, a CUDA framework for SMMO applications. \textsc{DynaSOAr} consists of (1) a fully-parallel, lock-free, dynamic memory allocator, (2) a data layout DSL and (3) an efficient, parallel do-all operation. \textsc{DynaSOAr} achieves performance superior to state-of-the-art GPU memory allocators by controlling both memory allocation and memory access.


\textsc{DynaSOAr} improves the usage of allocated memory with a Structure of Arrays (SOA) data layout and achieves low memory fragmentation through efficient management of free and allocated memory blocks with lock-free, hierarchical bitmaps. Contrary to other allocators, our design is heavily based on atomic operations, trading raw (de)allocation performance for better overall application performance. In our benchmarks, \textsc{DynaSOAr} achieves a speedup of application code of up to 3x over state-of-the-art allocators. Moreover, \textsc{DynaSOAr} manages heap memory more efficiently than other allocators, allowing programmers to run up to 2x larger problem sizes with the same amount of memory.



\end{abstract}

\section{Introduction}
General-purpose GPU computing has long been a tedious job, requiring programmers to write hand-optimized, low-level programs. In an attempt to make GPU computing available to a broader range of developers, our efforts are centered around bringing fast object-oriented programming (OOP) to low-level languages such as CUDA.

OOP has a wide range of applications in high-performance computing~\cite{bandini2009, Kale:1993:CPC:165854.165874, allan2010survey, 10.1007/978-3-540-25934-3_2, CARY199720} but is often avoided due to bad performance~\cite{master_th_patel}.
Dynamic memory management and the ability/flexibility of creating/deleting objects at any time is one of the corner stones of OOP. Due to the massive parallelism and data-parallel execution of GPUs, the number of simultaneous (de)allocations is significantly higher than on other parallel hardware architectures. In recent years, fast, dynamic memory allocators have been developed for GPUs~\cite{6339604,5577907, Widmer:2013:FDM:2458523.2458535, Vinkler:2015:RED:3071494.3071506, DBLP:journals/corr/abs-1710-11246, Spliet:2014:KDM:2588768.2576781, osti_1398234, Gelado:2019:TGM:3293883.3295727} and demanded by application developers~\cite{Zhu:2015:PIM:2817095.2817115, master_th_cuda_allc, doi:10.1002/cpe.3808, IJNC126, Schafer:2013:RLD:2492045.2492052, Li:2014:ENS:2701002.2701020, Li:2015:CAS:2769458.2769470}, showing a growing interest in better programming models and abstractions that have long been available on other platforms. However, while these allocators often provide good (de)allocation performance, they miss key optimizations for structured data, leading to poor data locality and memory bandwidth utilization when accessing allocated memory.

\paragraph*{Single-Method Multiple-Objects (SMMO)}
We identified a class of high-performance computing applications that can be expressed as object-oriented programs and implemented efficiently on SIMD architectures such as GPUs. We call this class \emph{Single-Method Multiple-Objects} (SMMO). The most fundamental operation of SMMO is parallel \emph{do-all}: Running one method in parallel on all existing objects of a type (\emph{object set}). Such operations fit perfectly with the data-parallel SIMD execution model of GPUs and can be implemented very efficiently. The main challenge lies is the fact that the object set is dynamic: Objects can be created and deleted in GPU code. The main contribution of our work is the design and implementation of a dynamic memory allocator that works well with SMMO applications and runs entirely on the GPU.

\begin{figure}[!t]
  \centering
  \includegraphics[width=0.75\textwidth]{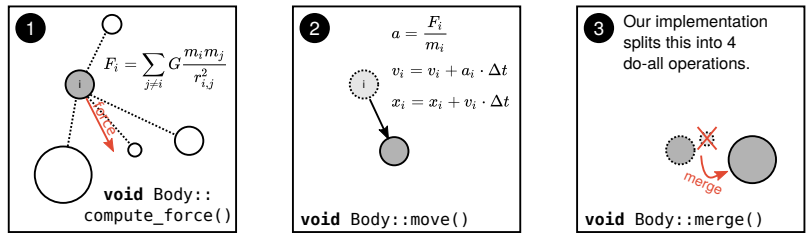}
  \caption{N-body Simulation with Collisions. The simulation consists of multiple do-all operations that are run in a loop for a fixed number of iterations (\emph{time steps})\protect\footnotemark. Every do-all operation runs in parallel and is a synchronization point: The next one can start only if the previous one has finished.}
  \label{fig:nbody_motivation}
\end{figure}
\footnotetext{We implement merging behavior with multiple do-all operations to avoid race conditions.}

\begin{figure}[!t]
  \centering
  \includegraphics[width=0.9\textwidth]{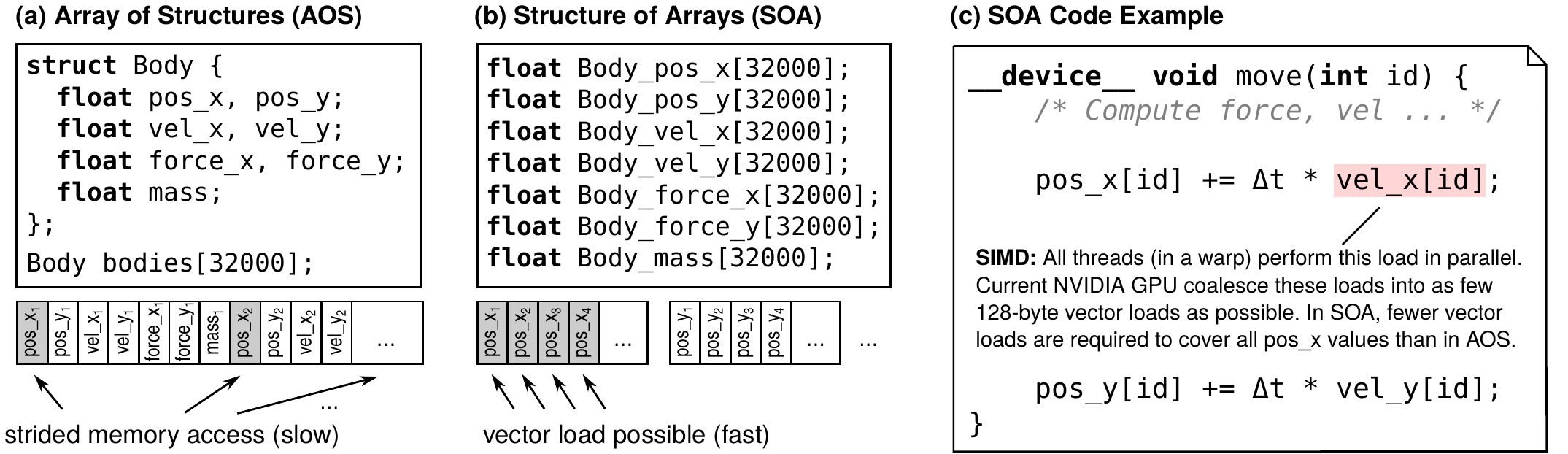}
  \caption{Data Layout of N-body Simulation in AOS and SOA. In SOA, multiple values of a field (e.g., \texttt{pos\_x$_1$} and \texttt{pos\_x$_2$}) can be loaded into a vector register with a single vector load instruction. In AOS, a less efficient, strided memory load or multiple smaller memory loads are necessary, because accessed data is not contiguous.}
  \label{fig:aos_soa_explain}
\end{figure}

SMMO is a broad class of problems with many real-world applications, such as social simulations~\cite{RePEc:mtp:titles:0262550253}, evacuation simulations~\cite{Li:2015:CAS:2769458.2769470}, predicting wildfire spreading~\cite{doi:10.1080/21580103.2016.1262793}, (adaptive~\cite{97c20025105249ca9f87087e9d7ec2c8}) finite element methods~\cite{FORDE1990355} or particle systems, to name just a few. As an example, consider the n-body simulation with collisions in Fig.~\ref{fig:nbody_motivation}. Such simulations are used by astronomers to simulate the collision of galaxies or the formation of planets~\cite{ALEXANDER1998113}. Every body is an object in SMMO and the simulation is a sequence of multiple do-all operations. 



\paragraph*{Structure of Arrays Data Layout}

Structure of Arrays (SOA) and Array of Structures (AOS) describe memory layouts for an object set~\cite{intel_aos_soa} (Fig.~\ref{fig:aos_soa_explain}). In AOS, the standard layout of most platforms, objects are stored as contiguous blocks of memory. In SOA, all values of a field are stored together. This allows for better cache utilization if not all fields are used in a computation. Moreover, it allows for efficient vector loads/stores on SIMD architectures. This is important, because SIMD architectures achieve parallelism by executing the same processor instruction on a \emph{vector} register. Previous work has reported speedups over AOS by multiple factors~(e.g.,~\cite{HOMANN2018325}).

Choosing the best data layout for an application is challenging and depends on the data access patterns of the application. Previous work has shown that a mixture of AOS and SOA can sometimes achieve the best performance~\cite{Franco:2017:YAG:3133850.3133861, 10.1007/978-3-662-48096-0_21, Weber:2017:MAL:3132652.3106341}. How to find good data layouts has been studied before~\cite{10.1007/978-3-662-48096-0_21, Abel:1999:ATS} and is out of the scope of this paper. We are focusing on SOA in this work, but \soaalloc{} could easily be extended to support other layouts in the future. Unfortunately, custom memory layouts come with a number of disadvantages:


\begin{description}
  \item[Missing OOP Abstractions] In a hand-written SOA layout, programmers refer to an object with an \emph{integer index} into SOA arrays (Fig.~\ref{fig:aos_soa_explain}\textsc{c}). However, OOP language abstractions (e.g., encapsulation, member access, method calls, type checking, inheritance) only work on object pointers/classes in mainstream languages. To overcome such issues, new languages (e.g., Shapes~\cite{Franco:2017:YAG:3133850.3133861}) and language dialects (e.g., ispc~\cite{6339601}) with built-in support for custom data layouts, as well as data layout libraries/DSLs for existing languages~\cite{STRZODKA2012429, Mattis:2015:COI:2814228.2814230, Springer:2018:ICD:3178433.3178439} have been developed.

  \item[Dynamic Object Set Size] SOA and AOS are not suitable for applications in which the number of objects changes over time, because programmers must specify a maximum object set size per type (e.g., 32,000 in Fig.~\ref{fig:aos_soa_explain}) ahead of time. Dynamic memory allocation solves this problem. As one of our contributions, we show how to allocate memory dynamically while preserving the performance characteristics of SOA. 


  \item[Subclassing/Inheritance] Inherited methods are shared between superclasses and subclasses. To allow a superclass method implementation to be used for a subclass, the subclass must use the same SOA arrays (and indices) as its superclass. In Columnar Objects, inherited SOA arrays are shared among all objects of all subclasses and newly introduced SOA arrays have a \texttt{null} value for objects of a super class~\cite{Mattis:2015:COI:2814228.2814230}. This approach works, but it can waste a considerable amount of memory.

\end{description}

\paragraph*{\textsc{DynaSOAr}: A Dynamic Allocator and C++/CUDA DSL for SOA Layout}
In this work, we present \soaalloc{}, a CUDA framework for SMMO applications. \soaalloc{} is a parallel, lock-free, dynamic memory allocator, combined with an efficient do-all operation and an embedded C++/CUDA DSL to enable OOP abstractions with custom object layouts. 

We are focusing on \soaalloc{}'s dynamic memory allocator and do-all operation in this work. \textsc{DynaSOAr} controls the data layout through its memory allocator and data access through its do-all operation. In SMMO applications, \soaalloc{} achieves superior performance compared to state-of-the-art allocators due to three main optimizations.

\begin{itemize}
\item Objects are stored in a \ul{Structure of Arrays (SOA)} data layout, a best practice for structured data in SIMD programs, making usage of allocated memory more efficient when used in conjunction with \soaalloc{}'s do-all operation.
\item Memory fragmentation caused by dynamic object (de)allocation is minimized with \ul{hierarchical bitmaps}. This is important because fragmentation diminishes the benefit of the SOA layout through less efficient vectorized access (more vector transactions are need to access fragmented data) and adversely affects cache performance~\cite{Grunwald:1993:ICL:155090.155107}.
\item Object allocation and deallocation performance is optimized with a number of \ul{low-level techniques}. For example, \soaalloc{} combines allocation requests within SIMD thread groups (\emph{warps}) to reduce the number of memory accesses during allocations~\cite{5577907} and takes advantage of efficient bit operations/intrinsics.
\end{itemize}






\paragraph*{Contributions and Outline}
This paper makes the following contributions.

\begin{itemize}
  \item The concept of Single-Method Multiple-Objects (SMMO) applications. We show that a variety of important HPC problems are SMMO applications.
  \item The design and implementation of \soaalloc{}, a dynamic object allocator for CUDA; with fast (de)allocation and a do-all operation. To the best of our knowledge, \textsc{DynaSOAr} is the first dynamic allocator that stores objects in an SOA data layout.
  \item An extension of the SOA data layout to dynamic object sets and subclassing.
  \item A concurrent, lock-free, hierarchical bitmap, based on atomic operations and retry loops.
  \item A comparison and evaluation of existing GPU memory allocators on SMMO applications.
\end{itemize}

The remainder of this paper is organized as follows. Sec.~\ref{sec:design_goals} gives an overview of the design goals of \soaalloc{}, focusing on memory access considerations of GPUs. Sec.~\ref{sec:arch_overview} describes the high-level architecture of \soaalloc{} and Sec.~\ref{sec:optimizations} explains important optimizations such as hierarchical bitmaps. Sec.~\ref{sec:related_work} compares the design of \soaalloc{} with other allocators and Sec.~\ref{sec:benchmark} evaluates application performance and fragmentation using microbenchmarks and multiple SMMO applications. Finally, Sec.~\ref{sec:conclusion} concludes the paper. Additionally, we provide a systematic correctness analysis in the appendix. 

\section{Design Goals}
\label{sec:design_goals}
\textsc{DynaSOAr} is a CUDA framework for SMMO applications and consists of three parts.

\begin{description}
  \item[Memory Allocator] We developed a dynamic memory allocator that provides \texttt{new}/\texttt{delete} operations in GPU code and stores objects in an SOA data layout. The main task of the allocator is to decide where to store each field value of each object on the heap.
  \item[Data Layout DSL] We developed an embedded C++ DSL to support OOP abstractions while storing objects in a custom layout. We could alternatively implement \textsc{DynaSOAr} in a language that allows programmers to specify custom data layouts (e.g., Shapes~\cite{Franco:2017:YAG:3133850.3133861, Tasos:2018:ESS:3242947.3242951} or ispc~\cite{6339601}), but such languages have limited GPU support.
  \item[Parallel Do-All] We developed an object enumeration strategy for SMMO applications that achieves efficient access of allocated memory on SIMD architectures. By controlling memory allocation and memory access, applications can achive better performance with \textsc{DynaSOAr} than with other state-of-the-art allocators, which are only concerned with memory allocation.
\end{description}

\textsc{DynaSOAr}'s DSL builds on top of Ikra-Cpp, an embedded C++ DSL for object-oriented programming with SOA layout~\cite{Springer:2018:ICD:3178433.3178439}. Its purpose is to make \textsc{DynaSOAr} easier to use for programmers. This paper is mainly about the memory allocator and the do-all operation.

\subsection{Programming Interface}
In contrast to general memory allocators, \soaalloc{} is an \emph{object allocator}. The types (classes/structs) that can be allocated must be specified at compile time. \soaalloc{} provides five basic operations. All operations except for \texttt{parallel\_do} and \texttt{parallel\_new} are \emph{device} functions that can only be called from GPU code.

\begin{itemize}
  \item \texttt{HAllocatorHandle::parallel\_do<T, \&T::func>(args...)}: Launches a GPU kernel that runs a member function \texttt{T::func} for all objects of type $T$ and subtypes\footnote{To avoid branch divergence, we launch a separate kernel for every type.} existing at launch time (\emph{parallel do-all}). \texttt{T::func} may allocate new objects, but those are not enumerated by the same parallel do-all operation. \texttt{T::func} may deallocate any object of different type $U \not= T$, but the object it is bound to (\texttt{\textbf{this}}) is the only object of type $T$ it may deallocate (delete itself). This is to avoid race conditions.
  \item \texttt{HAllocatorHandle::parallel\_new<T>(n, args...)}: Launches a GPU kernel that instantiates $n$ objects of type $T$. In addition to \texttt{args...}, the constructor receives an ID $i$ between 0 and $n - 1$ (for the $i$\textsuperscript{th} object) as the first argument.
  \item \texttt{\textbf{new}(d\_allocator) T(args...)}: Allocates a new object of type $T$ and returns a pointer to the object. The \emph{placement new} notation~\cite{cpp_placement} is a common C++ pattern for arena allocation and \texttt{d\_allocator} is the allocator/arena in which the object is allocated.
  \item \texttt{destroy(d\_allocator, ptr)}: Deletes an object that was allocated with \texttt{d\_allocator}\footnote{There is no \emph{placement delete} syntax, so it is a common pattern to provide a separate \texttt{destroy} function~\cite{placement_delete}.}.
  \item \texttt{DAllocatorHandle::device\_do<T, \&T::func>(args...)}: Runs a member function \texttt{T::func} for all objects of type $T$ in the current GPU thread. Can only be used inside of a \texttt{parallel\_do} or a manually launched GPU kernel. This is a sequential \emph{for-each} loop. It is typically used for processing all pairs of objects (e.g., in n-body simulations). 
\end{itemize}

Listing~\ref{lst:short_example} shows parts of the n-body simulation of Fig.~\ref{fig:nbody_motivation} to illustrate \textsc{DynaSOAr}'s API and DSL. 

\begin{figure}
\begin{lstlisting}[caption={\textsc{DynaSOAr} API Example: Excerpt from an n-body simulation with collisions.}, label={lst:short_example}]
#include "dynasoar.h"

class Body;  // Pre-declare all classes. This simple example has only one class.
using AllocatorT = SoaAllocator</*max_num_obj=*/ 16777216, /*T...=*/ Body>;
__device__ DAllocatorHandle<AllocatorT> d_allocator;

class Body : public AllocatorT::Base {  // Can subclass other user-defined class.
 public:
  // Pre-declare all field types. DynaSOAr uses these to compute the size of blocks.
  declare_field_types(Body, float /*pos_x_*/, float /*pos_y_*/,
                            /* ... */, bool /*was_merged_*/)

 private:
  // Declare fields with proxy types but use like normal C++ fields (as in Ikra-Cpp).
  Field<Body, 0> pos_x_;               // Position X
  Field<Body, 1> pos_y_;               // Position Y
  /* other fields omitted... */
  Field<Body, 9> was_merged_;          // Was this body merged into another one?

 public:
  __device__ Body(float pos_x, float pos_y, float vel_x, float vel_y, float mass)
    : pos_x_(pos_x), pos_y_(pos_y), vel_x_(vel_x), vel_y_(vel_y), mass_(mass) {}

  // This constructor is invoked by parallel_new.
  __device__ Body(int idx)
      : Body(/*pos_x=*/ random_float(-kMaxPos, kMaxPos),
             /*pos_x=*/ random_float(-kMaxPos, kMaxPos), /* ... */) {}

  __device__ void apply_force(Body* other) {
    if (other != this) {
      float dx = pos_x_ - other->pos_x_;  float dy = pos_y_ - other->pos_y_;
      float dist = sqrt(dx*dx + dy*dy);
      float F = kGravityConstant * mass_ * other->mass_ / (dist * dist);
      other->force_x_ += F * dx / dist;  other->force_y_ += F * dy / dist;
    }
  }

  __device__ void step_1_compute_force() {
    force_x_ = force_y_ = 0.0f;
    d_allocator->device_do<Body, &Body::apply_force>(this);
  }

  __device__ void step_2_move(float dt) {
    vel_x_ += force_x_ * dt / mass_;  vel_y_ += force_y_ * dt / mass_;
    pos_x_ += dt * vel_x_;            pos_y_ += dt * vel_y_;
  }

  __device__ void step_6_delete_merged() {
    if (was_merged_) { destroy(d_allocator, this); }
  }
};

int main() {
  // Create new allocator. This will allocate a large buffer ("heap") on the GPU.
  auto* h_allocator = new HAllocatorHandle<AllocatorT>();
  // Copy device handle to d_allocator handle.
  cudaMemcpyToSymbol(d_allocator, h_allocator->device_handle(),
                     cudaMemcpyHostToDevice);  // a bit simplified...

  // Create 65536 random body objects. We do not use the new keyword in this example.
  // Alternatively, we could run this in a kernel: new(d_allocator) Body(...)
  h_allocator->parallel_new<Body>(65536);

  for (int i = 0; i < kIterations; ++i) {
    h_allocator->parallel_do<Body, &Body::step_1_compute_force>();
    h_allocator->parallel_do<Body, &Body::step_2_move>(/*dt=*/ 0.5);
    /* some steps omitted... */
    h_allocator->parallel_do<Body, &Body::step_6_delete_merged>();
  }

  delete h_allocator;  // Deallocate buffer and all allocations within.
  return 0;
}
\end{lstlisting}
\end{figure}

\subsection{Memory Access Performance}
The main insight of our work is that optimizing only for fast (de)allocations is not enough. Optimizing the access of allocated memory can result in much higher speedups, because device (\emph{global}) memory access is the biggest bottleneck of memory-bound GPU applications:

\begin{description}
\item[Latency] Global memory access instructions have a very high latency at around 400--800 clock cycles, compared to arithmetic instructions at around 6--24 cycles. Programmers can hide latency with \emph{high occupancy}~\cite{Volkov:EECS-2016-143} (i.e., running many threads).

\item[Memory Bandwidth] The global memory bandwidth is a limiting factor. Peak memory transfer rates can be achieved only with \emph{memory coalescing}: When the threads in a GPU application simultaneously access different memory addresses, the GPU coalesces accesses from the same SIMD thread group (\emph{warp} in CUDA, every 32 consecutive threads) into one physical transaction if the addresses are on the same 128-byte cache line~\cite{5473222}. However, if threads access data on multiple cache lines (e.g., non-contiguous, spread-out addresses), more transactions are needed\footnote{This is similar to vectorized loads/stores, but coalescing is performed by the hardware.}, which reduces transfer rates significantly. The CUDA Best Practices Guide puts a \emph{high priority} note on coalesced memory accesses~\cite{nvidia_memoryco}. 

\item[Caches] Hits in the L1/L2 cache are served much faster (less latency, memory bandwidth pressure) than global memory loads. Field reordering and structure splitting are common techniques for increasing the number of hot fields in cache~\cite{Chilimbi:1999:CSD:301618.301635}.

\end{description}

\textsc{DynaSOAr} achieves good memory access performance with a SOA-style data layout: First, SOA increases memory coalescing because values of the same field, which are accessed simultaneously in SIMD, are stored together. Second, SOA is an extreme form of structure splitting and can improve cache utilization because fields that are not accessed do not occupy cache lines.


\subsection{High Density Memory Allocation}
A SOA data layout (Fig.~\ref{fig:clustering_ex}a) achieves good memory performance but is not suitable for dynamic allocation: The size of SOA arrays is fixed and new allocations cannot be accommodated once all array slots are occupied.

\textsc{DynaSOAr}'s design is based on the insight that a \emph{clustered layout} with SOA-style structure splitting (Fig.~\ref{fig:clustering_ex}b) has the same cache/vector performance characteristics as a SOA layout, if scalar values are stored in dense clusters of at least 128~bytes (vector and cache line size) and clusters are aligned to 128~bytes, regardless of where the clusters are located. This gives \textsc{DynaSOAr} more freedom in the placement of allocations and is exploited by its allocation policy.


\begin{figure}[!t]
  \centering
  \includegraphics[width=\textwidth]{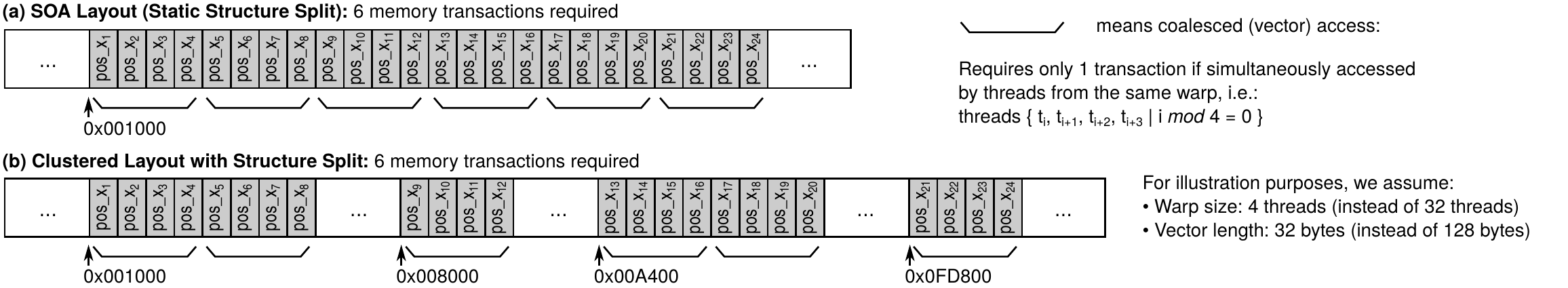}
  \caption{Data Layouts: Number of required memory transactions to read 24 floats simultaneously.}
  \label{fig:clustering_ex}
\end{figure}

\subsection{Parallel Object Enumeration Strategy}
Current GPUs follow the Single-Instruction Multiple-Threads (SIMT) execution model. Intuitively, every SIMD lane corresponds to a thread and every group of consecutive 32 threads forms a \emph{warp} which executes the same instruction on a vector register.

To benefit from memory coalescing, the threads of a warp must access addresses on the same 128-byte L1 cache line. In a SOA data layout, this is achieved when the threads of a warp read/write the same fields of objects with \ul{contiguous indices} at the same time. Intuitively, threads in a warp should process \emph{neighboring} (spatially local) objects.

In \textsc{DynaSOAr}, programmers invoke GPU kernels with parallel do-all operations. These operations must (a) spawn enough GPU threads to hide latency, but not too many to avoid inefficiencies, and (b) assign objects to threads in such a way that memory access is optimized.



\subsection{Scalability}
Memory allocations require some sort of synchronization between threads to prevent \emph{collisions}, i.e., two threads allocating the same memory location. To avoid collisions, some allocators such as Cilk~\cite{Blumofe:1999:SMC:324133.324234} utilize private heaps, but such designs can lead to high memory consumption (\emph{blowup})~\cite{Berger:2000:HSM:378993.379232} and are in feasible on massively parallel architectures with thousands of threads.

State-of-the-art GPU allocators such as ScatterAlloc~\cite{6339604} and Halloc~\cite{hallocweb} reduce collisions with hashing, which scatters allocations almost randomly on the heap. This would render a SOA layout useless and defeat one of \textsc{DynaSOAr}'s main optimizations.

With such design restrictions, \textsc{DynaSOAr} is bound to have less efficient allocations than other allocators. However, as we show throughout this paper, \textsc{DynaSOAr} can more than make up for slow allocations with more efficient memory access.


Previous CPU memory allocator designs emphasize mechanisms for reducing false sharing, which can degrade performance~\cite{Berger:2000:HSM:378993.379232}. This is not an issue on GPUs, because L1 caches are not coherent. Programmers must use the \texttt{volatile} keyword or atomic operations to enfore a read/write to the shared L2 cache or global memory. 



\section{Architecture Overview}
\label{sec:arch_overview}
\begin{figure}
  \centering
  \includegraphics[width=\columnwidth]{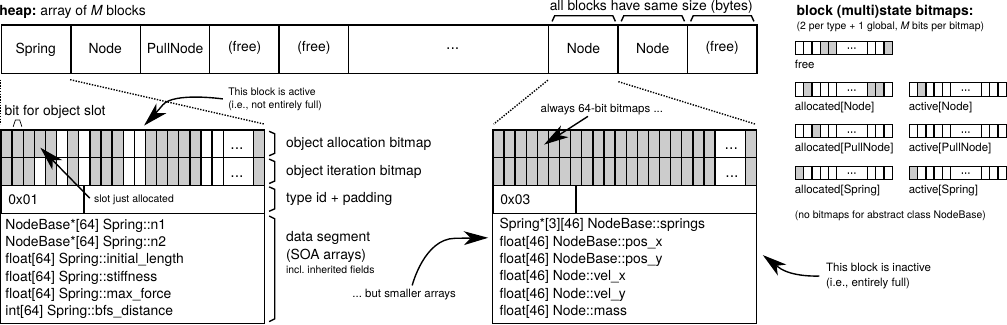}
  \caption{Example: Heap layout for a FEM simulation of a crack in a composite material. The heap is divided into $M$ blocks of equal size. Every block has the same structure: an allocation bitmap, an iteration bitmap, and a type identifier, followed by a data segment storing objects in SOA layout.}
  \label{fig:overview_heap_soaalloc}
\end{figure}
\begin{figure}
  \begin{minipage}[c]{0.4\textwidth}
    \includegraphics[width=\textwidth]{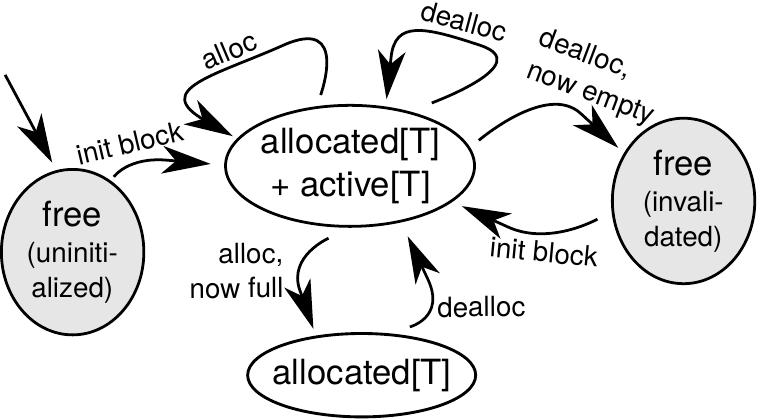}
  \end{minipage}\hfill
  \begin{minipage}[c]{0.56\textwidth}
    \caption{Block State Transitions. At first, blocks are in an uninitialized state. As part of allocation, new active blocks may be initialized (\emph{allocated}). Active blocks become inactive when they are full. Inactive blocks become active again an object is deallocated. Active blocks are invalidated when their last object is deallocated. Invalidated blocks can be reinitialized (to any type) and are handled similar to uninitialized blocks.
    } \label{fig:block_states}
  \end{minipage}
\end{figure}

\soaalloc{} manages a single, large heap in global memory on device. The heap is divided into $M$ blocks of equal number of bytes. $M$ is determined at compile time based on the block size. Multiple objects of the same type (C++ class/struct) are stored in a block in a Structure of Arrays (SOA) data layout (Fig.~\ref{fig:overview_heap_soaalloc}). Once a block is initialized (\emph{allocated}) for a certain type, only objects of that type can be stored in that block until the block (and all its objects) is deallocated again and reinitialized to a different type.

The maximum number of objects in a block depends on its type, because different structs/classes may have different sizes. To improve clustering, \soaalloc{} allocates new objects in already existing, non-full blocks (\emph{fast path}). We call such blocks \emph{active}, because they participate in allocations (Fig.~\ref{fig:block_states}). Only if no active block could be found, a new block is allocated and becomes active (\emph{slow path}).

\subsection{Block Structure}
Every block has two 64-bit object bitmaps: An \emph{object allocation bitmap} and an \emph{object iteration bitmap}. The allocation bitmap tracks allocated slots in the block. The iteration bitmap is used for object enumeration and overwritten with the allocation bitmap before every parallel do-all operation. Its purpose is to ensure that objects that were created during a do-all operation are not enumerated by the same do-all operation; that would a race condition.

The \emph{type identifier} is a unique ID for the type $T$ of a block. The remainder of the block is occupied by padding and the \emph{data segment}, storing $1 \leq N_T \leq 64$ objects in SOA layout. The data segment begins with SOA arrays for inherited fields and ends with SOA arrays of newly introduced fields.

Slots are marked as (de)allocated with atomic AND/OR operations that change a single bit of the object allocation bitmap. Based on their return value\footnote{An atomic operation returns the value in memory before modification.}, we know ...

\begin{itemize}
  \item ... if an allocation was successful or another thread was faster allocating the same slot.
  \item ... if a particular allocation filled up a block (i.e., allocated the last slot).
  \item ... if a particular deallocation emptied a block (i.e., deallocated the last slot).
\end{itemize}

If a thread filled up a block or emptied a block, it is that thread's \emph{responsibility} to update the other internal data structures. This is a common pattern in lock-free designs~\cite{Michael:2004:SLD:996841.996848}. Note that every block has the same byte size and structure; e.g., the bitmaps are always at the same offset. This is an important property for the correctness of our lock-free (de)allocation algorithms and simplifies \emph{safe memory reclamation}.

\subsection{Block Capacity}
The capacity of a block (maximum number of objects) depends on the size (bytes) of the type of objects in the block. If \soaalloc{} manages objects of types $T_1$, $T_2$, ..., $T_n$ and $s=\argmin_{i \in 1...n} \mathit{size}(T_i)$ is the index of the smallest type, then the capacity $N_{T}$ of a block of type $T$ is determined as follows.
\begin{align*}
N_{T} = \left\lfloor{\frac{64 \cdot \mathit{size}(T_s)}{\mathit{size}(T)}}\right\rfloor \tag{\emph{block capacity}}
\end{align*}

A block of the smallest type $T_s$ has capacity 64. Given a fixed heap size, the size of $T_s$ determines the block size in bytes and thus the number of blocks $M$.

As soon as a type $T$ is more than twice as big as $T_s$, the benefit of the SOA layout starts fading away for $T$, because $N_T < 32$. The maximum amount of memory coalescing can only be achieved with vector loads (cluster sizes) of 32 values (assuming 32-bit scalar types). Furthermore, \soaalloc{} cannot handle cases in which a type is more than 64 times bigger than the smallest type. In reality, these limitations proved to be insignificant. None of our benchmarks experienced a slowdown due to unfavorable block sizes.


\subsection{C++ Data Layout DSL and Object Pointers}
\begin{figure}
  \begin{minipage}[c]{0.5671195351024146\textwidth}
    \includegraphics[width=\textwidth]{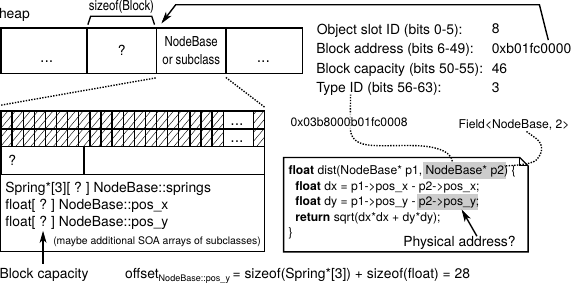}
  \end{minipage}\hfill
  \begin{minipage}[c]{0.4\textwidth}
    \caption{Object Pointer Example. The static type of \texttt{p2} is \texttt{NodeBase*}. The corresponding block has SOA arrays for \texttt{NodeBase} fields and for the additional fields of the runtime type of \texttt{p2}. The size of those arrays is not statically known and depends on the runtime type of \texttt{p2}.
    } \label{fig:global_ref}
  \end{minipage}
\end{figure}

Field access is simple in most object-oriented systems: Given an object pointer, which is a memory location, a field value is stored at a fixed offset from the object pointer.

%


In \textsc{DynaSOAr}, an object pointer is not a memory location, but a combination of various components (\emph{fake pointer}~\cite{Springer:2018:ICD:3178433.3178439}), similar to \emph{global references} in \emph{Shapes}~\cite{Franco:2017:YAG:3133850.3133861}. Upon field access, the \textsc{DynaSOAr} DSL transparently converts object pointers to memory locations, without breaking C++'s OOP abstractions. We follow the implementation strategy of Ikra-Cpp, where fields are declared with proxy types \texttt{Field<B, N>}, which are implicitly converted to \texttt{T\&} values~\cite{cpp_obj}, where \texttt{T} is the N-th predeclared field type of \texttt{B}~\cite{Springer:2018:ICD:3178433.3178439}. This conversion is defined by our DSL and computes the actual, physical memory location within a data segment.

A \textsc{DynaSOAr} object pointer (Fig.~\ref{fig:global_ref}) is based on the address of the block in which the object is located. All blocks are aligned to 64~bytes, so we can store the object slot ID in the 6 least significant bits. Since recent GPU architectures have at most 24~GB of memory and no virtual memory, only the 35 least significant bits are used in memory addresses and the remaining 29 bits are always zero\footnote{We experimentally verified this on NVIDIA Maxwell and NVIDIA Pascal.}. We store additional information in these bits: The 8 most significant bits store the type identifier for fast instance-of checks. The next 6 bits store the capacity of the block. Note that, while C++ stores runtime types with a vtable pointer at the beginning of an object, we store runtime type information in unused pointer bits.

While in most object-oriented systems, runtime type information is only required for virtual function calls, \textsc{DynaSOAr} needs the block capacity (a property of the runtime type) also for field accesses, because SOA array offsets within the data segment depend on it.

For example, \texttt{p2} in Fig.~\ref{fig:global_ref} is statically known to be of type \texttt{NodeBase*}, but the block capacity (size of SOA arrays) depends on the runtime type, which can be any subclass of \texttt{NodeBase}. Those subclasses can have different block capacities. The size of SOA arrays and the object slot ID are required to compute the physical location of \texttt{p2->pos\_y}, so we encode both inside object pointers.




This computation, along with bit-shifting and bit-AND operations for extracting all components from an object pointer, is performed on every field read/write (Sec.~\ref{sec:address_computation_ex}). This overhead may seem large, but arithmetic operations are much faster than memory access, even in case of an L2 cache hit. Overall, the performance benefit of SOA is much larger than the address computation overhead.


\subsection{Block Bitmaps}
To find blocks or free memory quickly during object enumeration or object allocation, \soaalloc{} maintains three bitmaps of size $M$, where $M$ is the maximum number of blocks on the heap.

\begin{itemize}
  \item The \emph{free block bitmap} indexes block locations that are not yet allocated. This bitmap is used to determine where new blocks are allocated. Bit $i$ is 1 iff block $i$ is free (uninitialized or invalidated). Initially, every bit is 1.
  \item There is one \emph{block allocation bitmap} for every type $T$. That bitmap indexes blocks of type $T$ and is used for enumeration of all objects. Blocks of subclasses are not included in bitmaps of the superclass. Initially, every bit is 0.
  \item There is one \emph{active block bitmap} for every type $T$, indexing allocated, non-full blocks. If a bit is 1, then the same bit in the block allocation bitmap must also be 1. This bitmap is used to find a block in which a new object can be allocated. Initially, every bit is 0.
\end{itemize}

Due to concurrent (de)allocations, block bitmaps cannot be kept consistent with the actual block states all the time, as indicated by object allocation bitmaps and type identifiers of blocks. However, we designed our algorithms in such a way that they can handle such inconsistencies and keep block states and block bitmaps \emph{eventually consistent}.

\subsection{Object Slot Allocation}
When a new object is created, \textsc{DynaSOAr} allocates memory and runs the constructor on the object pointer. Alg.~\ref{alg:alloc_algo} shows how memory is allocated. This algorithm runs entirely on the GPU and is completely lock-free.

\textsc{DynaSOAr} tries to allocate memory in an already existing, active block. If no block could be found, it first initializes a new block at a location that is known to be free (\emph{slow path}). The state of the new block is \emph{allocated} and \emph{active}, so that the new block can also be found by other threads.

Once a block was selected, an object slot is reserved by atomically finding and flipping a bit from 0 to 1 in the object allocation bitmap (details in Alg.~\ref{alg:block_allocate}). Based on the return value of the atomic operation, we know if this operation just allocated the last slot. In that case, the block is marked as \emph{inactive} in the active block bitmap (Line~12).

Since the allocator is used concurrently by many threads, we may select a block (Line~2) that is full or no longer exists when attempting to reserve an object slot (Line~8). If the block is full, object reservation fails and we retry by selecting a new active block. If the block no longer exists, we have to consider three cases\footnote{We give a more systematic correctness argument in the appendix.}.

\begin{enumerate}
  \item There is currently no block at this location. In this case, object reserveration fails, because all slots are marked as allocated in the object allocation bitmap when a block is deleted. We call this process \emph{block invalidation}.
  \item The block was deleted and a new block of the same type was allocated at the same location. Such ABA problems are harmless and allocation will succeed.
  \item The block was deleted and there is now a block of different type at the same location. At this point, the constructor has not run yet, so no data in the data segment was corrupted. This is because all blocks have the same structure, i.e., the object allocation bitmap is always at the same location. We can safely rollback the allocation by running the deallocation routine.
\end{enumerate}




\SetKwComment{Comment}{$\triangleright$\ }{}
\SetAlgoVlined
\begin{algorithm}[t]
\small
 \Repeat(\Comment*[f]{\textsf{Infinite loop if OOM}}){false}{
  bid $\gets$ active[T].\emph{try\_find\_set}(); \hfill \Comment{\textsf{Find and return the position of any set bit.}}
  \If(\Comment*[f]{\textsf{Slow path}}){\emph{bid} = FAIL} {
    bid $\gets$ free.\emph{clear}();   \hfill \Comment{\textsf{Find and clear a set bit atomically, return position.}}
    \emph{initialize\_block}<T>(bid);  \hfill \Comment{\textsf{Set type ID, initialize object bitmaps.}}
    allocated[T].\emph{set}(bid)\;
    active[T].\emph{set}(bid)\;
  }
  alloc $\gets$ heap[bid].\emph{reserve}();  \hfill \Comment{\textsf{Reserve an object slot. See Alg.~\ref{alg:block_allocate}.}}
  \If{\emph{alloc} $\not=$ FAIL}{
    ptr $\gets$ \emph{make\_pointer}(bid, alloc.slot)\;
    t $\gets$ heap[bid].type; \hfill \Comment{\textsf{Volatile read}}
    \lIf{\emph{alloc.state} = FULL}{
      active[t].\emph{clear}(bid)
    }
    \lIf{t = T}{
      \Return ptr
    }{
      \emph{deallocate}<t>(ptr); \hfill \Comment{\textsf{Type of block has changed. Rollback.}}
    }
  }
 }
 \caption{DAllocatorHandle::allocate<T>() : T* \hfill \fbox{GPU}}
 \label{alg:alloc_algo}
\end{algorithm}
\SetKwComment{Comment}{$\triangleright$\ }{}
\begin{algorithm}[t]
\small
\begin{multicols}{2}
  bid $\gets$ \emph{extract\_block}(ptr)\;
  slot $\gets$ \emph{extract\_slot}(ptr)\;
  state $\gets$ heap[bid].\emph{deallocate}(slot)\;
  \uIf{\emph{state} = FIRST}{
    active[T].\emph{set}(bid)
  }
  \ElseIf{\emph{state} = EMPTY}{
    \If{invalidate(bid)}{
      t $\gets$ heap[bid].type\;
      active[t].\emph{clear}(bid)\;
      allocated[t].\emph{clear}(bid)\;
      free.\emph{set}(bid)\;
    }
  }
  \end{multicols}
\vspace{0.25cm}
 \caption{DAllocatorHandle::deallocate<T>(T* ptr) : void \hfill \fbox{GPU}}
 \label{algo:dealloc_a}
\end{algorithm}

\subsection{Object Deallocation}
When an object is deleted, \textsc{DynaSOAr} extracts its runtime type $T$ from the object pointer. Then, \textsc{DynaSOAr} runs the C++ destructor and deallocates the memory (Alg.~\ref{algo:dealloc_a}) as follows.

We first extract block and object slot IDs from the object pointer and free the object slot by atomically flipping its bit in the object allocation bitmap from 1 to 0. Based on the return value of the atomic operation we know the fill level of the block right before the deallocation.

If this deallocation freed the first object slot (block previously full), we mark the block as active (Line~5), so that other threads can find it and allocate objects in it. If this deallocation freed the last object slot (block now empty), we attempt to delete the block (Lines~7--11). Safe memory reclamation is known to be difficult in lock-free algorithms~\cite{Michael:2002:SMR:571825.571829}. The main problem is that one or more contending threads, in the course of their lock-free operations, may have selected the block that we are about to delete for new allocations.

To avoid the block from being modified by other threads, we \emph{invalidate} it. Block invalidation attempts to atomically flipping all bits in the object allocation bitmap from 0 to 1. If this atomic operation failed to flip at least one bit from 0 to 1 (because it was already 1), another thread must have reserved an object slot in the meantime. In this case, we rollback the changes to the object allocation bitmap and abort block invalidation and deletion.

If invalidation was successful, the block is guaranteed to be empty and cannot be modified by other threads anymore because all bits in the object allocation bitmap are 1. The type of the block may have changed in the meantime (Line~8), but it is now safe to mark this block location as \emph{free}, so that a new block can be initialized at this location.




\subsection{Parallel Object Enumeration: \texttt{parallel\_do}}
\label{sec:par_do_all_sec3}
Parallel do-all is the foundation of SMMO applications. It launches a GPU kernel that runs a method \texttt{T::func} on all objects of a type $T$. That method may read and write fields of the object that it is bound to (\texttt{this}). The goal of parallel do-all is to assign objects to GPU threads in such a way that memory coalescing is maximized for those field accesses.

\begin{figure}
  \begin{minipage}[c]{0.52877004242\textwidth}
    \includegraphics[width=\textwidth]{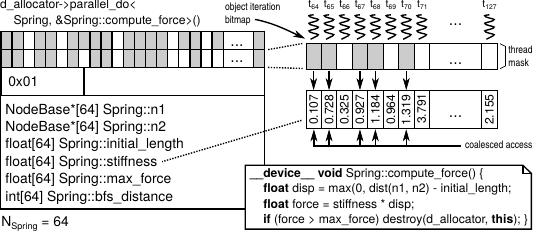}
  \end{minipage}\hfill
  \begin{minipage}[c]{0.43\textwidth}
  \caption{Thread Assignment Example. 64 threads with consecutive IDs (2 warps) are assigned to every allocated block of type \texttt{Spring}. Since not all object slots are in use, as indicated by the block iteration bitmap, some threads have no work to do. All other threads can benefit from memory coalescing when reading/writing fields of the object that they are assigned to.}
  \label{fig:thread_assignment_ex}
  \end{minipage}
\end{figure}

Memory coalescing is maximized when all threads of a warp access consecutive memory addresses at the same time. In this case, all those memory accesses can be serviced by efficient vector loads/writes. In CUDA, threads are identified by thread IDs. Each warp consists of a consecutive range of threads. E.g., warp 0 consists of threads $t_0, t_1, \ldots t_{31}$. Assuming a block capacity of $N_T$, \soaalloc{} assigns $N_T$ consecutive threads to the objects in a block (Fig.~\ref{fig:thread_assignment_ex}). This leads to good memory coalescing on average. Perfect memory coalescing can be achieved if the following two conditions apply.

\begin{itemize}
  \item $N_T$ is a multiple of the warp size 32. If this is not the case, then there are warps whose threads process elements in two or more different blocks at the same time.
  \item Objects have good clustering, i.e., every block except for at most one is entirely full. Due to the way objects are allocated (only in active blocks), we expect a high fill level.
\end{itemize}

\soaalloc{} uses the block allocation bitmap to find blocks to which threads should be assigned. Assigning only one object to a thread is too inefficient if the number of objects is large. Therefore, a thread $t_\mathit{tid}$ may have to process an object slot in multiple blocks. Our scheduling strategy always assigns the same object slot position $\mathit{id}_O(\mathit{tid})$, but in multiple blocks $\mathit{id}_B(\mathit{tid})$ (strided by the number of threads~\cite{cuda_grid_stride}), to a thread. In those formulas, $R$ is an array of indices of all allocated blocks of type $T$, i.e., all blocks containing objects of type $T$. The total number of threads $n$ can be hand-tuned by the programmer. With those formulas, every thread can by itself determine the objects that it should process.

\begin{figure}
  \begin{minipage}[c]{0.44187973217\textwidth}
    \includegraphics[width=\columnwidth]{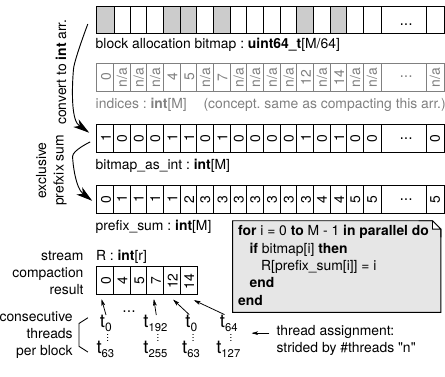}
  \end{minipage}\hfill
  \begin{minipage}[c]{0.558\textwidth}
\begin{align*}
\mathit{id}_O(\mathit{tid}) = \mathit{tid} \,\,\, \% \,\,\, N_T \tag{\emph{assigned obj. slot idx.}}
\end{align*}
\begin{align*}
\mathit{id}_B(\mathit{tid}) = \left(\hspace{-0.1cm}R\hspace{-0.1cm}\left[\frac{\mathit{tid} + k \cdot n}{N_T}\right] \left.\right\vert k \in [0; \mathit{num}_B(\mathit{tid}))\hspace{-0.1cm}\right) \tag{\emph{assigned block indices}}
\end{align*}
\begin{align*}
\mathit{num}_B(\mathit{tid}) =  \left\lceil \frac{r\cdot N_T - \mathit{tid}}{n} \right\rceil \tag{\emph{number of assigned blocks}}
\end{align*}
  \end{minipage}
      \caption{Example: Compacting block allocation bitmap indices and assigning $n=256$ threads to 6 allocated blocks with $N_T = 64$. This prefix sum-based implementation retains the order of indices (i.e., $R$ is sorted), but this is not necessary for correctness.
    } \label{fig:scan_enumeration}
\end{figure}

The array $R$ is required because every thread must by itself find the $\frac{\mathit{tid}}{N_T}$-th, $\frac{\mathit{tid}+n}{N_T}$-th, etc. allocated block of type $T$ quickly, without scanning the entire block allocation bitmap. \textsc{DynaSOAr} precomputes $R$ before every parallel-do operation (Fig.~\ref{fig:scan_enumeration}). Conceptually, this is an application of stream compaction~\cite{IJNC151} and usually implemented with a prefix sum~\cite{Sengupta06awork-efficient, Billeter:2009:ESC:1572769.1572795}: Given a bitmap of size $M$, generate an \emph{indices} array of size $M$ containing $i$ at position $i$ if the $i$-th bit is set. Otherwise, store an \emph{invalid marker}. Now filter/compact the array to retain only valid values, resulting in an array $R$ of size $r$. Note that we do not care if the original ordering of indices is retained. Sec.~\ref{sec:hierarchical_bit} describes how this algorithm is further optimized with hierarchical bitmaps to avoid scanning empty bitmap parts.



\section{Optimizations}
\label{sec:optimizations}
This section describes performance optimizations that \soaalloc{} applies in addition to the SOA data layout to achieve good (de)allocation performance.

\subsection{Hierarchical Bitmaps}
\label{sec:hierarchical_bit}
\soaalloc{} uses bitmaps for finding blocks or free space for blocks. Since, with growing heap sizes, bitmaps can reach several megabytes in size, we use a hierarchy of bitmaps, such that \emph{set} bits (ones) can be found with a logarithmic order of memory accesses.

Our hierarchical bitmaps are structurally recursive (i.e., bitmaps nested in each other) and hide their hierarchy as an implementation detail from their interface. Such bitmaps are used in database systems~\cite{10.1007/978-3-540-39403-7_19} and garbage collectors~\cite{Ueno:2011:ENG:2034773.2034802}, but we do not know of any hierarchical bitmaps that support concurrent modifications.

\subsubsection{Data Structure}
A hierarchical bitmap of size $N$ bits consists of two parts: an array of size $\lceil N/64\rceil$ of 64-bit \emph{containers} (\texttt{uint64\_t}), and a \emph{nested bitmap} of size $\lceil N/64 \rceil$ if $N > 64$. A container $C_i^l$ consists of bits $b_{64 \cdot i}^l$, ...,  $b_{64 \cdot i + 63}^l$ and is represented by one bit $b_{i}^{l+1}$ in the nested (higher-level) bitmap (Fig.~\ref{fig:bitmap_clear_example}). That bit is set if at least one bit is set in the container.

\begin{align*}
b_i^{l+1} = \bigvee_{k=0}^{63} b_{64\cdot i + k}^{l} \tag{\emph{container consistency}}
\end{align*}

We chose a container size of 64 bits because C++ has a 64-bit integer type and CUDA (and most other architectures) provide atomic operations for modifying 64-bit values. Bits in a container are changed with atomic operations. Higher-level bits (and thus bitmaps) are \emph{eventually consistent} with their containers. Keeping both consistent all the time is impossible without locking, because two different memory locations cannot be changed together atomically. However, due to the design of the bitmap operations, the bitmap is guaranteed to be in a consistent state when all bitmap operations (of all threads) are completed, at the end of a GPU kernel. Bitmap operations retry or give up (\emph{FAIL}) to handle temporary inconsistencies. This is a key difference compared to other lock-free hierarchical data structures such as SNZI~\cite{Ellen:2007:SSN:1281100.1281106}, which have stronger runtime consistency guarantees and require more complex algorithms.

\subsubsection{Operations}
\label{sec:bitmap_operations_412}
\begin{figure}[t]
  \begin{minipage}[c]{0.4\textwidth}
    \includegraphics[width=\columnwidth]{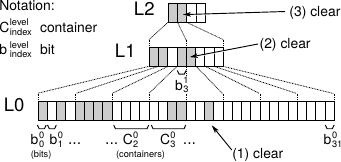}
  \end{minipage}\hfill
  \begin{minipage}[c]{0.55\textwidth}
    \caption{Example: Hierarchical bitmap of size 32 with container size 4 (instead of 64). This example illustrates how (1) a \emph{clear}(18) operation triggers (2) a \emph{clear}(4) operation in the nested bitmap, which triggers (3) a \emph{clear}(1) operation in the next nested bitmap.} \label{fig:bitmap_clear_example}
  \end{minipage}
\end{figure}
All bitmap operations except for \emph{indices()} are device functions that run entirely on the GPU. All operations that modify memory are thread-safe and their semantics are atomic. Internally, they are all implemented with atomic memory operations. 

\begin{itemize}
  \item $\texttt{try\_clear}(\texttt{pos})$ atomically sets the bit at position $\texttt{pos}$ to 0. Returns \emph{true} if the bit was 1 before and \emph{false} otherwise.
  \item $\texttt{clear}(\texttt{pos})$ \emph{switches} the bit at position $\texttt{pos}$ from 1 to 0. Retries until the bit was actually changed by the current thread. This is identical to \texttt{\textbf{while} (!try\_clear(pos)) \{\}}.

  \item $\texttt{set}(\texttt{pos})$ switches the bit at position $\texttt{pos}$ from 0 to 1. Retries until the bit was changed.
  \item $\texttt{try\_find\_set}()$ returns the position of an arbitrary bit that is set to 1 or \emph{FAIL} if none was found. Must be used with caution, because the returned bit position might already have changed when using the result.
  \item $\texttt{clear}()$ atomically clears and returns the position of an arbitrary set bit. This is identical to \texttt{\textbf{while} ((i = try\_find\_set()) != FAIL \&\& try\_clear(i)) \{\}; \textbf{return} i;}.
  \item $\texttt{get}(\texttt{pos})$ returns the value of the bit at position $\texttt{pos}$.
  \item $\texttt{indices}()$ returns an array of indices of all set bits. This is a host function and cannot be used in a GPU kernel.
\end{itemize}

\SetKwComment{Comment}{$\triangleright$\ }{}
\begin{algorithm}[t]
\small
\begin{multicols}{2}
 cid $\gets$ pos / 64\;
 offset $\gets$ pos \% 64\;
 mask $\gets$ 1 {<}{<} offset\;
 prev $\gets$ \emph{atomicAnd}(\&container[cid], $\sim$mask)\;
 success $\gets$ (prev \& mask) $\not=$ 0\;

\If{\hspace{-4pt} \narrowstyle success $\wedge$ has\_nested $\wedge$ {popc}(\emph{prev})\hspace{-0.17cm}\tikz[remember picture] \node [] (d){};\hspace{-2pt} \emph{=\,1}}{
  nested.\emph{clear}(cid)\;
}
\begin{tikzpicture}[remember picture, overlay,
    every edge/.append style = { ->, thick, >=stealth,
                                  dashed, line width = 1pt },
    every node/.append style = { align = center,
                                 font=\sffamily\small, fill= gray!20},
                  text width = 1.9cm ]
  \node [notice={(-0.2,0.15)}, below right = 0.3cm and -1.5cm of d]  (D) {\mbox{\begin{varwidth}{1.9cm}\narrowstyle \scriptsize \textbf{population cnt.}: number of set bits\end{varwidth}}};

\end{tikzpicture} \hspace{-0.17cm} \textbf{return } success;
\end{multicols}
 \caption{Bitmap::try\_clear(pos) : void  \hfill \fbox{GPU}}
 \label{alg:clear}
\end{algorithm}

\begin{algorithm}[t]
\small
\begin{multicols}{2}
\uIf{has\_nested}{
  cid $\gets$ nested.\emph{try\_find\_set}()\;
  \lIf{\emph{cid} = FAIL}{
    \Return \emph{FAIL}
  }
}\Else{
  cid $\gets$ 0\;
}
offset $\gets$ \emph{ffs}\tikz[remember picture] \node [] (d){};\hspace{-0.17cm}(container[cid])\;
\eIf{offset = \emph{NONE}}{
  \Return \emph{FAIL}
}{
\begin{tikzpicture}[remember picture, overlay,
    every edge/.append style = { ->, thick, >=stealth,
                                  dashed, line width = 1pt },
    every node/.append style = { align = center,
                                 font=\sffamily\small, fill= gray!20},
                  text width = 1.79cm ]
  \node [notice={(-0.3,0.075)}, below right = 0.0cm and 2.3cm of d]  (D) {\mbox{\begin{varwidth}{1.79cm}\narrowstyle \scriptsize \textbf{find first set}: idx. of 1\textsuperscript{st} set bit\end{varwidth}}};

\end{tikzpicture}
  \Return 64*cid + offset\;
}

\end{multicols}
 \caption{Bitmap::try\_find\_set() : int \hfill \fbox{GPU}}
 \label{alg:top_down_tr}
\end{algorithm}
\begin{algorithm}[t]
\small
\begin{multicols}{2}
\eIf{has\_nested}{
  selected $\gets$ nested.\emph{indices}()
}{
  selected $\gets$ [0]
}
R $\gets$ array(N)\;
r $\gets$ 0\;
\For({\,\,\,$\triangleright$ \fbox{GPU}}){$\mbox{cid} \in \mbox{selected}$ \emph{\textbf{in parallel}}}{
  c $\gets$ container[cid]\;
  s $\gets$ \emph{atomicAdd}(\&r, \emph{popc}(c))\;
  \For{$i \gets 0$ \KwTo \emph{\emph{popc}(c))}}{
    R[s + i] $\gets$ 64*cid + \emph{nth\_bit}\tikz[remember picture] \node [] (d){};\hspace{-0.17cm}(c, i)\;
  }
}
 \begin{tikzpicture}[remember picture, overlay,
    every edge/.append style = { ->, thick, >=stealth,
                                  dashed, line width = 1pt },
    every node/.append style = { align = center,
                                 font=\sffamily\small, fill= gray!20},
                  text width = 1.2cm ]
  \node [notice={(-0.15,0.125)}, below left = 0.25cm and -0.9cm of d]  (D) {\mbox{\begin{varwidth}{1.2cm}\narrowstyle \scriptsize idx. of $i$\textsuperscript{th} set bit in $c$\end{varwidth}}};

\end{tikzpicture} \hspace{-0.17cm} \Return R.\emph{subarray}(0, r)\;
\end{multicols}
 \caption{Bitmap::indices() : int[N] \hfill \fbox{CPU}}
 \label{alg:indices}
\end{algorithm}

\subsubsection{Set and Clear with Atomic Operations}
As many other lock-free algorithms, our hierarchical bitmaps are based on a combination of atomic operations and retries~\cite{doi:10.1002/9781119332015.ch3}. The return value of an atomic operation indicates if a bit was actually changed and if it is this thread's responsibility to update the higher-level bitmap (Fig.~\ref{fig:bitmap_clear_example}).

As an example, Alg.~\ref{alg:clear} shows how to clear the bit at position \texttt{pos}. In Line~4, the respective container is bit-ANDed with a mask containing ones everywhere except for that position. This will clear the bit at position \texttt{pos} but leave all other bits unchanged. The current thread actually changed the bit if it is set in \texttt{prev} (Line~5). If this operation cleared the last bit (Line~6), then the bit in the higher-level bitmap must be cleared.

Note that higher-level bits are always changed with \emph{clear(pos)}/\emph{set(pos)} and not with their respective \emph{try\_} versions, because other concurrently running bitmap operations that are still in process may not have updated all higher-level bitmaps yet, leaving the data structure in a temporarily inconsistent state. If we were to use \emph{try\_} versions, a mandatory update of the higher-level bitmap could be accidentally dropped due to a bitmap inconsistency. \emph{clear(pos)}/\emph{set(pos)} ensure that the update is performed eventually by retrying (and spin-blocking the thread) until the update was successful.

\subsubsection{Finding an Arbitrary Set Bit}
Instead of scanning the entire L0 bitmap, set bits can be found faster with a top-down traversal of the bitmap hierarchy, as shown in Alg.~\ref{alg:top_down_tr}. A request is first delegated to the higher-level bitmap (Line~2) to select a container. When that call returns, a set bit is chosen in the selected container (Line~6).

Even if the bitmap has set bits, this operation can fail if it reads an inconsistent combination of containers from different hierarchy levels. For example, consider that a container with exactly one set bit is chosen by the recursive call. However, before reaching Line~6, another thread clears that bit as part of a concurrent bitmap operation. In that case, \emph{try\_find\_set} fails even though there may be set bits in other containers.

\textsc{DynaSOAr}'s performance is affected by such bitmap inconsistencies when searching for active blocks (Alg.~\ref{alg:alloc_algo}, Line~2). While bitmap inconsistencies do not affect correctness, they lead to higher fragmentation because \textsc{DynaSOAr} will initialize additional blocks even though objects could be accommodated in already existing blocks. We analyze the effect of such bitmap inconsistencies in our benchmarks (Sec.~\ref{sec:detailed_analysis_wator_s}).

\subsubsection{Enumerating Set Bit Indices}
\label{sec:enumerating_set_bit_indices}
Before launching a parallel do-all kernel, \soaalloc{} uses the \emph{indices} operation to generate a compact array of allocated block indices ($R$ in Fig.~\ref{fig:scan_enumeration}). No GPU code is running at this time, so the bitmap is guaranteed to be in a consistent state. To ensure good scaling with increasing heap sizes, and thus increasing block bitmap sizes, \soaalloc{} utilizes the bitmap hierarchy to quickly skip containers without any set bits (Alg.~\ref{alg:indices}).

First, an index array is generated for the higher-level bitmap (Line~2). This array is then processed in parallel; the \emph{for} loop in Line~7 is a GPU kernel and every thread processes one or multiple containers selected by the recursive call. If a container $C_i^l$ does not have any set bits, then its corresponding bit $b_i^{l+1}$ is in a cleared state in the higher-level bitmap and not included in \texttt{selected}. Every thread reserves space in the result array $R$ by increasing an atomic counter and fills its portion of the array with bit indices. This algorithm proved to be faster and requires less memory than a prefix sum algorithm, which needs multiple array copies/buffers per bitmap. Note that, in contrast to the prefix sum-based implementation of Sec.~\ref{sec:par_do_all_sec3}, this algorithm does not retain the order of indices and $R$ and is not sorted.


\subsection{Reducing Thread Contention}
\label{sec:less_allocation_cont}
In Alg.~\ref{alg:top_down_tr} and~\ref{alg:block_allocate}, threads are competing with each other for bits: Only one thread can reserve any given object slot and only a limited number of threads can succeed with allocations in a block. To guarantee correctness, our design is heavily based on atomic operations. These operations became considerably faster with recent GPU architectures~\cite{DeGonzalo:2019:AGW:3314872.3314884, warp_aggre}, but performance can still suffer when too many threads choose the same bit, because threads have to retry if allocation fails. \soaalloc{} employs two techniques to reduce such thread contention. 

\begin{description}
\item [Allocation Request Coalescing] Originally proposed by XMalloc~\cite{5577907}, \soaalloc{} combines memory allocation requests of the same type within a warp. One \emph{leader thread} reserves all object slots in a single block on behalf of all participating threads (optimized version of Alg.~\ref{alg:block_allocate}). If the selected active block does not have enough free object slots, \soaalloc{} reserves as many slots as possible and then chooses another active block for the remaining allocation requests. This reduces atomic memory operations, because multiple bits in an object allocation bitmap are set in one operation. Furthermore, the constructor for newly allocated objects can run more efficiently, because field accesses are coalesced.

\item [Bitmap Rotation] Instead of a plain \emph{find first set} (ffs) in Alg.~\ref{alg:top_down_tr} and~\ref{alg:block_allocate}, bitmaps are first rotating-shifted by a value depending on the warp ID and a seed that is changed with every retry. This increases the probability of threads choosing different active blocks for allocation and reduces the probability of threads trying to reserve the same object slots in a block. This is a key optimization technique that improved performance by an order of magnitude. 
\end{description}

While bitmap traversals are relatively cheap, block initializations are expensive because in addition to initializing object bitmaps, bits in three different bitmaps (plus hierarachy) must be changed (slow path of Alg.~\ref{alg:block_allocate}). To avoid unnecessary block initializations, it proved beneficial to retry the search for active blocks (Line~2) a constant number of times before entering the slow path. This optimization resulted in lower fragmentation and improved performance.


\subsection{Efficient Bit Operations}
\SetKwComment{Comment}{$\triangleright$\ }{}
\begin{algorithm}[t]
\small
\begin{multicols}{2}
  \Repeat{\emph{success} $\vee$ \emph{block\_full}} {
    pos $\gets$ \emph{ffs}($\sim$bitmap)\;
    \lIf{\emph{pos} = NONE}{
      \Return \emph{FAIL}
    }
    mask $\gets$ 1 {<}{<} pos\;
    before $\gets$ \emph{atomicOr}(\&bitmap, mask)\;
    success $\gets$ (before \& mask)) = 0\;
    block\_full $\gets$ before = 0xFF...F\;
  }
  \If{{success}}{
    \eIf{popc(\emph{before}) \emph{= 63}}{
      \Return (pos, \emph{FULL})
    }{
      \Return (pos, \emph{REGULAR})
    }
  }
\Return \emph{FAIL}\;
\end{multicols}
\vspace{0.1cm}
 \caption{Block::reserve() : (int, state) \hfill \fbox{GPU}}
 \label{alg:block_allocate}
\end{algorithm}

\soaalloc{} is taking advantage of efficient bitwise operations such as \emph{ffs} (``find first set'') and \emph{popc} (``population count''). Modern CPU and GPU architectures have dedicated instructions for such operations. As an example, Alg.~\ref{alg:block_allocate} shows how a single object slot is reserved. Instead of checking all bits in a loop, \emph{ffs} in Line~2 is used to find a free slot (index of a cleared bit) in the object allocation bitmap and \emph{popc} in Line~10 counts the number of previously allocated slots (number of set bits) to decide if this request filled up the block.

As another example, due to allocation request coalescing, every thread must now extract its reserved object slot from a set of allocations performed by a leader thread on behalf of the entire warp. This boils down to finding the $i$-th set bit in a 64-bit bitmap $b$ of newly reserved object slots, where $i$ is the rank of a thread among all allocating threads in the warp. Instead of checking every bit in $b$ one-by-one (loop with 64 iterations in the worst case) and keeping track of the number of set bits seen so far, we apply \emph{b $\gets$ b \& (b - 1)} in a loop $i - 1$ times (to clear the first $i - 1$ bits) and then calculate \emph{ffs}(\emph{b}). We omit the details of this optimization here, as it is only one example for a variety of similar low-level optimizations.



\section{Related Work}
\label{sec:related_work}
CUDA provides an on-device dynamic memory allocator, but it is unoptimized and slow. To solve this issue, multiple custom allocators have been developed over the last years. These allocators achieve good performance by exploiting an allocation pattern that many applications on massively parallel SIMD architectures exhibit: Most allocations are small in size and due to mostly regular control flow, many allocations have the same byte size.

Halloc~\cite{hallocweb} is one of these allocators. It is a slab allocator and can allocate only a few dozen predetermined byte sizes between 16~bytes and 3~KB. This is fast but can lead to internal fragmentation. \textsc{DynaSOAr} can avoid such internal fragmentation because allocation sizes are determined from compile-time type information of the application. A slab in Halloc contains same-size allocations and tracks allocations with a bitmap. To avoid scanning large bitmaps, a hash function determines which bits to check during allocations. Only one slab can be active per allocation size and if the active slab becomes too full, it is replaced with a new one. In contrast, more than one block per type can be active in \textsc{DynaSOAr} and blocks are filled up entirely.

XMalloc~\cite{5577907} is the first allocator with allocation request coalescing, which was adopted by many other allocators, including \textsc{DynaSOAr}. Coalesced requests are served from \emph{basicblocks}, which are organized in one of multiple lock-free free lists depending on their size.

FDGMalloc maintains a private heap for every warp~\cite{Widmer:2013:FDM:2458523.2458535}, similar to Hoard~\cite{Berger:2000:HSM:378993.379232}. It does not have a general \emph{free} operation and can only deallocate entire heaps, so it is not suitable for SMMO applications.

CircularMalloc (CMalloc)~\cite{Vinkler:2015:RED:3071494.3071506} allocates memory in a ring buffer. Every allocation has a pointer to the next allocation or free chunk, wrapping around at the end of the buffer. CMalloc traverses the linked list for free chunks during allocations. To reduce allocation contention, every multiprocessor starts its traversal at a different location. This is similar to \textsc{DynaSOAr}'s bitmap rotation optimization.

ScatterAlloc~\cite{6339604} hashes allocation requests to memory \emph{pages} depending on their allocation size and the multiprocessor ID. Pages hold allocations of the same size, but slightly smaller requests can be accommodated, leading to internal fragmentation. While \soaalloc{} uses hierarchical bitmaps, ScatterAlloc uses hashing with linear probing for finding pages during allocations. For benchmarks, we use mallocMC~\cite{eckert_carlchristian_helmut_johannes_2014_34461}, a reimplementation of ScatterAlloc that is still maintained.


Both Halloc and ScatterAlloc maintain fill levels to quickly skip congested memory areas that are above a certain threshold, because the performance of any hashing technique degrades with an increasing number of collisions. In \soaalloc{}, temporary inconsistencies in bitmap hierarchies increase with the number of concurrent allocations, but \soaalloc{} can dynamically adapt to such cases by initializing additional blocks.


\begin{table}[!htp]
\caption{Description of Benchmark Applications}
\label{tab:smmo_apps}
\setlength\tabcolsep{4pt}
\begin{tabularx}{\textwidth}{cXllllll}
\hline\hline
 \narrowstyle & \footnotesize \textbf{Benchmark Description} & \footnotesize \rotatebox[origin=l]{90}{\narrowstyle\textbf{\#par. do-all}} & \footnotesize \rotatebox[origin=l]{90}{\textbf{\narrowstyle\#classes}\,\,\,\,} & \footnotesize \rotatebox[origin=l]{90}{\narrowstyle\textbf{alloc./dealloc.}} & \footnotesize \rotatebox[origin=l]{90}{\narrowstyle\textbf{smallest class\,\,\,\,}} & \footnotesize \rotatebox[origin=l]{90}{\narrowstyle\textbf{largest class}} \\
\Xhline{2\arrayrulewidth}
\narrowstyle \raisebox{-0.85\totalheight}{\includegraphics[width=0.175\textwidth]{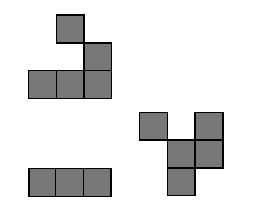}} &
\footnotesize \textbf{Game of Life:} A cellular automaton due to J. H. Conway. This version has a time complexity of O(\#alive cells) instead of the standard O(\#cells) algorithm. Cells can be dead, alive or alive-candidates. Alive-candidates are dead cells that may become active in the next iteration. Only alive-candidates and alive cells are processed with parallel do-all operations.
& \footnotesize\rotatebox[origin=r]{90}{4 / iteration} & \footnotesize\rotatebox[origin=r]{90}{4 (2 dyn.)} & \footnotesize\rotatebox[origin=r]{90}{\ding{51} / \ding{51}} & \footnotesize\rotatebox[origin=r]{90}{5B, 2 fields} & \footnotesize\rotatebox[origin=r]{90}{8B, 1 field} \\
\hline
\narrowstyle\raisebox{-0.9\totalheight}{\includegraphics[width=0.175\textwidth]{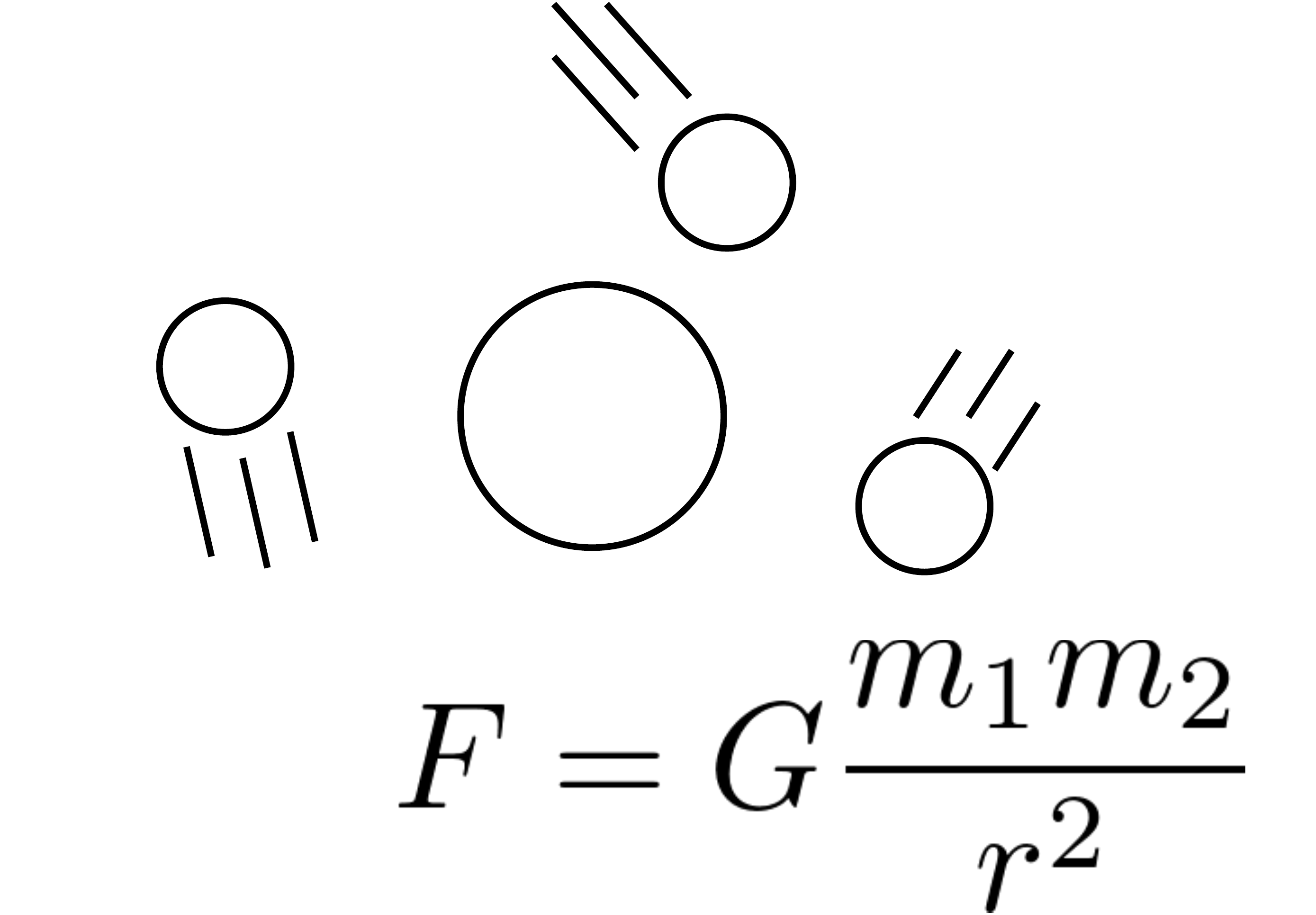}} & \footnotesize \textbf{N-Body:} Simulates the movement of particles according to gravitational forces. A \texttt{device\_do} operation is required to calculate (and then sum up) the gravitational force between every pair of particles. All objects are allocated upfront. This benchmark has no dynamic object (de)allocation. & \footnotesize\rotatebox[origin=r]{90}{2 / iteration}  & \footnotesize \rotatebox[origin=r]{90}{1 (0 dyn.)} & \footnotesize\rotatebox[origin=r]{90}{\ding{55} / \ding{55}} & \footnotesize\rotatebox[origin=r]{90}{\,\,\,28B, 7 fields} & \footnotesize\rotatebox[origin=r]{90}{(same)} \\
\hline
\narrowstyle\raisebox{-0.9\totalheight}{\includegraphics[width=0.175\textwidth]{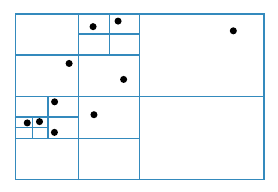}} & \footnotesize \textbf{Barnes-Hut:} An extension of N-Body in which bodies are stored in a quad tree~\cite{BURTSCHER201175} (2D), to evaluate \textsc{DynaSOAr} with dynamic tree data structures. The running time is dominated by the construction/maintenance (i.e., frequent node inserts/removals) of the quad tree via parallel top-down/bottom-up tree traversals. Tree nodes are dynamically (de)allocated.
& \footnotesize\rotatebox[origin=r]{90}{10 / iteration} & \footnotesize\rotatebox[origin=r]{90}{3 (1 dyn.)} & \footnotesize\rotatebox[origin=r]{90}{\ding{51} / \ding{51}} & \footnotesize\rotatebox[origin=r]{90}{68B, 9 fields} & \footnotesize\rotatebox[origin=r]{90}{\,\,\,102B, 12 fields}\\
\hline
\raisebox{-0.95\totalheight}{\includegraphics[width=0.175\textwidth]{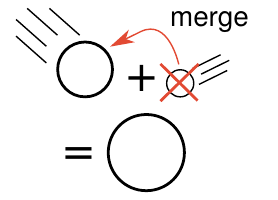}} & \footnotesize \textbf{Particle Collisions:} Similar to N-Body, but particles are merged according to perfectly inelastic collision when they are getting too close. The number of particles decreases gradually. This benchmark has dynamic object deallocation but no dynamic object allocation. & \footnotesize\rotatebox[origin=r]{90}{6 / iteration} & \footnotesize\rotatebox[origin=r]{90}{1 (1 dyn.)} & \footnotesize\rotatebox[origin=r]{90}{\ding{55} / \ding{51}} & \footnotesize\rotatebox[origin=r]{90}{\,\,\,38B, 10 fields} & \footnotesize\rotatebox[origin=r]{90}{(same)} \\
\hline
\raisebox{-0.88\totalheight}{\includegraphics[width=0.175\textwidth]{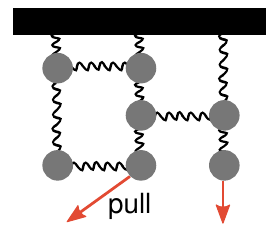}} &\footnotesize \textbf{Structure:} Simulates a fracture in a composite material, modeled as a FEM. Intuitively, the mesh is a graph and edges between nodes are springs. When pulling the mesh on one side, the material starts to break eventually. Isolated nodes are detected with a BFS~\cite{Harish:2007:ALG:1782174.1782200} and removed. Literature describes extensions that would benefit from dynamic allocation~\cite{LU2018240}. & \footnotesize\rotatebox[origin=r]{90}{3 / iteration} & \footnotesize\rotatebox[origin=r]{90}{5 (4 dyn.)} & \footnotesize\rotatebox[origin=r]{90}{\ding{55} / \ding{51}} & \footnotesize\rotatebox[origin=r]{90}{32B, 6 fields} & \footnotesize\rotatebox[origin=r]{90}{46B, 7 fields} \\
\hline
\raisebox{-0.85\totalheight}{\includegraphics[width=0.175\textwidth]{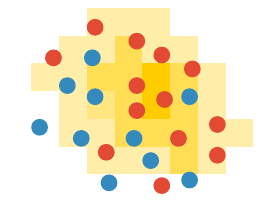}} & \footnotesize \textbf{Sugarscape:} An agent-based social simulation~\cite{RePEc:mtp:titles:0262550253}. Agents inhabit a 2D grid and can move to neighboring cells. Cells contain sugar which is consumed by agents. Sugarscape can simulate a variety social dynamics (e.g., trade, war, environmental pollution). Our simulation is quite simple. We simulate resource consumption, ageing and mating.  & \footnotesize\rotatebox[origin=r]{90}{12 / iteration} & \footnotesize\rotatebox[origin=r]{90}{4 (2 dyn.)} &  \footnotesize\rotatebox[origin=r]{90}{\ding{51} / \ding{51}} & \footnotesize\rotatebox[origin=r]{90}{52B, 7 fields} & \footnotesize\rotatebox[origin=r]{90}{\,\,\,74B, 11 fields} \\
\hline
\raisebox{-0.85\totalheight}{\includegraphics[width=0.175\textwidth]{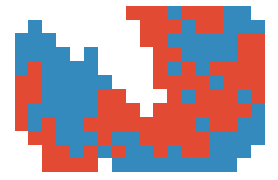}} & \footnotesize \textbf{Wa-Tor:} An agent-based predator-prey simulation~\cite{10.2307/24969495}. Fish/sharks occupy a 2D grid of cells and can move to neighboring cells. Fish and sharks reproduce after some iterations. Fish die when they are eaten and sharks starve to death when they run out of food. & \footnotesize\rotatebox[origin=r]{90}{8 / iteration} & \footnotesize\rotatebox[origin=r]{90}{4 (2 dyn.)} & \footnotesize\rotatebox[origin=r]{90}{\ding{51} / \ding{51}} & \footnotesize\rotatebox[origin=r]{90}{60B, 4 fields} & \footnotesize\rotatebox[origin=r]{90}{\,\,\,64B, 5 fields} \\
\hline
\raisebox{-0.9\totalheight}{\includegraphics[width=0.175\textwidth]{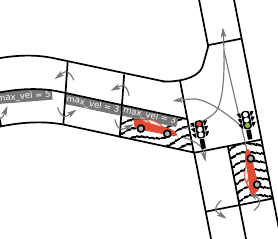}} & \footnotesize \textbf{Nagel-Schreckenberg:} A traffic flow simulation on a street network~\cite{nagel_schr}. This simulation can reproduce traffic jams and other real-world traffic phenomena. Streets are modeled as a network of cells, with at most one vehicle per cell. New vehicles are continuously added to the simulation and existing vehicles are removed at their final destination. & \footnotesize\rotatebox[origin=r]{90}{3 / iteration}  & \footnotesize\rotatebox[origin=r]{90}{4 (1 dyn.)} & \footnotesize\rotatebox[origin=r]{90}{\ding{51} / \ding{51}} & \footnotesize\rotatebox[origin=r]{90}{97B, 10 fields} & \footnotesize\rotatebox[origin=r]{90}{124B, 6 fields}\\
\hline
& \footnotesize \textbf{Linux Scalability:} Not an SMMO application. This microbenchmark allocates, then deallocates a fixed number of same-size objects in each thread, without ever accessing the memory~\cite{Lever:2000:MPM:1267724.1267780}. & \footnotesize \rotatebox[origin=r]{90}{n/a} & \footnotesize\rotatebox[origin=r]{90}{1 (1 dyn.)} & \footnotesize\rotatebox[origin=r]{90}{\ding{51} / \ding{51}} & \footnotesize\rotatebox[origin=r]{90}{\,\,\,4B, 1 field} & \footnotesize\rotatebox[origin=r]{90}{(same)} \\
\hline\hline
\end{tabularx}
\setlength\tabcolsep{6pt}
\end{table}

\section{Benchmarks}
\label{sec:benchmark}
We evaluated \soaalloc{} with multiple real-world SMMO applications that exhibit different memory allocation patterns (Table~\ref{tab:smmo_apps}). We ran all benchmarks on a computer with an Intel Core i7-5960X CPU, 32~GB main memory and an NVIDIA TITAN Xp GPU (12~GB device memory), and compiled them with nvcc (-O3) from the CUDA Toolkit~9.1 on Ubuntu~16.04.4.

We compare the running time with different allocators. If possible, we also measured the running time of \emph{baseline} implementations that do not use any dynamic memory management.

\paragraph*{Benchmark Applications}
We describe all benchmarks and their implementation in detail (incl. their SMMO structure) on GitHub\footnote{\url{https://github.com/prg-titech/dynasoar/wiki/Benchmark-Applications} (also see artifact)}. Our benchmarks are from different domains and fall into four categories.

\begin{enumerate}
  \item Objects allocated up front, no deallocation: \textsf{nbody}
  \item Objects allocated up front, then only deallocation: \textsf{collision}, \textsf{structure}
  \item Cellular automaton (CA) with static cells network: \textsf{sugarscape}, \textsf{traffic}, \textsf{wa-tor}
  \item Other: \textsf{barnes-hut}, \textsf{game-of-life}
\end{enumerate}

Baselines (SOA/AOS) are application variants without any dynamic memory allocation. Baselines of category (1) are trivial to implement with static allocation. In category (2), every object has a boolean \texttt{active} flag to prevent deleted objects from being enumerated in the future. In category (3), classes are merged with the underlying static cell data structure, which wastes memory in case of empty cells (Sec.~\ref{sec:bench_mem_usage_sec}). Category (4) applications cannot be implemented with static allocation, unless the application is changed fundamentally.

\paragraph*{Parallel do-all in Custom Allocators}
Other allocators do not provide do-all operations, which are required for SMMO applications. To compare \soaalloc{} with other allocators, we developed standalone \texttt{parallel\_do} and \texttt{device\_do} implementations that can be used with any allocator.


These components maintain arrays for allocated and deleted objects of each type. Pointers are inserted into these arrays with atomic operations. At the end of a parallel do-all operation, deleted pointers are removed from the array of allocated pointers. Then, the array of allocated pointers is compacted with a prefix sum operation (same as Fig.~\ref{fig:scan_enumeration}).

Depending on the number of (de)allocations, this mechanism may take a long time. A better allocator-specific mechanism could likely be developed with some reverse engineering. For that reason, we show the amount of time spent on \emph{parallel enumeration}. This time should not be taken into account when comparing the performance of different allocators.

\paragraph*{BitmapAlloc}
To analyze the performance of pure bitmap-based object allocation without SOA layout, blocks and fake pointers, we developed a second allocator \emph{BitmapAlloc}. This allocator treats the entire heap as one large object array, whose slots are managed by hierarchical bitmaps, similarly to \soaalloc{}: one \emph{allocation bitmap} per type and one \emph{free slot bitmap}. Allocation bitmaps are also used for \texttt{parallel\_do} and \texttt{device\_do}.

The main downside of BitmapAlloc is its inefficient memory usage. It supports only a single allocation size, potentially leading to high internal fragmentation.

\subsection{Performance Overview}
Fig.~\ref{fig:bench_overview} shows the running time of all benchmarked SMMO applications. \textsc{DynaSOAr} achieves superior performance over other allocators due to the SOA layout, a dense object allocation policy and an efficient parallel do-all operation.

All applications except for \textsf{structure} see a speedup by switching from AOS to SOA (compare baselines). In \textsf{structure}, most fields are used together, so SOA does not pay off.

Despite having no dynamic (de)allocation during the benchmark, \textsf{nbody} can see a slight speedup with dynamic memory allocation. This is likely due to fewer cache associativity collisions compared to a denser allocation in array~\cite{7853809}.

In \textsf{collision}, \soaalloc{}/BitmapAlloc enumerate objects with a bitmap scan of the object allocation bitmap (\texttt{device\_do}; 1 bit/object), more efficently than other allocators. Other allocators read objects pointers from an array (8 bytes/object). The baseline versions read an \texttt{active} flag (1 byte/object) from every object, including deleted ones. 

\textsf{game-of-life} and \textsf{wa-tor} are applications that (de)allocate a large number of objects, so enumeration takes a long time. \soaalloc{} and BitmapAlloc have much more efficient parallel-do operations than other allocators.

\textsf{sugarscape} and \textsf{wa-tor} exhibit a 2D grid structure of cells. Baseline versions take advantage of this geometric structure, leading to more coalesced memory access, while programmers have no control over where objects are placed in memory by dynamic allocators. For this reason, the baseline versions are faster than the versions with dynamic memory management. 

In general, in applications with dynamic memory management, objects are referred to with 64-bit object pointers, while all baseline versions use 32-bit integer indices. This penalizes especially benchmarks with small objects; their object sizes grow considerably just by switching from 32-bit integers indices to 64-bit pointers.

\begin{table}
\caption{Comparison of Allocators. \emph{Coal.} means \emph{Allocation Request Coalescing}.}
\label{fig:baselines}
\begin{tabularx}{\columnwidth}{Xcccc}
\hline \hline
 \footnotesize \narrowstyle \textbf{Allocator} & \footnotesize \narrowstyle \textbf{Coal.} & \footnotesize \textbf{SOA} & \footnotesize \textbf{Container} & \footnotesize \textbf{Finding Free Memory} \\
\Xhline{2\arrayrulewidth}
 \footnotesize \narrowstyle \soaalloc{} & \footnotesize\ding{51} & \footnotesize\ding{51} & \footnotesize Block & \footnotesize Hierarchical Bitmap \\ \hline
 \footnotesize\narrowstyle \soaalloc{}-NoCoal & \footnotesize \ding{55} & \footnotesize \ding{51} & \footnotesize Block & \footnotesize \footnotesize Hierarchical Bitmap \\ \hline
 \footnotesize \narrowstyle BitmapAlloc & \footnotesize\ding{55} & \footnotesize\ding{55} & \footnotesize\ding{55} & \footnotesize Hierarchical Bitmap \\ \hline
 \textcolor{gray}{\footnotesize\narrowstyle CircularMalloc} & \textcolor{gray}{\footnotesize \ding{55}} & \textcolor{gray}{\footnotesize \ding{55}} & \textcolor{gray}{\footnotesize \ding{55}} & \textcolor{gray}{\footnotesize Linked List, Ring Buffer} \\ \hline
 \footnotesize\narrowstyle Default CUDA Allocator & \footnotesize \ding{55} & \footnotesize \ding{55} & \footnotesize (Unknown) & \footnotesize (Unknown) \\ \hline
 \textcolor{gray}{\footnotesize \narrowstyle FDGMalloc} & \textcolor{gray}{\footnotesize \ding{51}} & \textcolor{gray}{\footnotesize \ding{55}} &  \textcolor{gray}{\footnotesize Priv. Heap, Superblock} &  \textcolor{gray}{\footnotesize Linked List}\\ \hline
 \footnotesize \narrowstyle Halloc & \footnotesize \ding{55} & \footnotesize \ding{55} & \footnotesize Slab & \footnotesize Bitmap, Hashing \\ \hline
 \footnotesize \narrowstyle mallocMC (ScatterAlloc) & \footnotesize \ding{51} & \footnotesize \ding{55} & \footnotesize Superblock, Region, Page & \footnotesize Hashing \\ \hline
 \textcolor{gray}{\footnotesize \narrowstyle XMalloc} & \textcolor{gray}{\footnotesize \ding{51}} & \textcolor{gray}{\footnotesize \ding{55}} & \textcolor{gray}{\footnotesize (4 block hierarchies)} & \textcolor{gray}{\footnotesize Lock-free Free Lists}  \\
 \hline
\hline
\end{tabularx}
\end{table}

\begin{figure}
  \includegraphics[width=\textwidth]{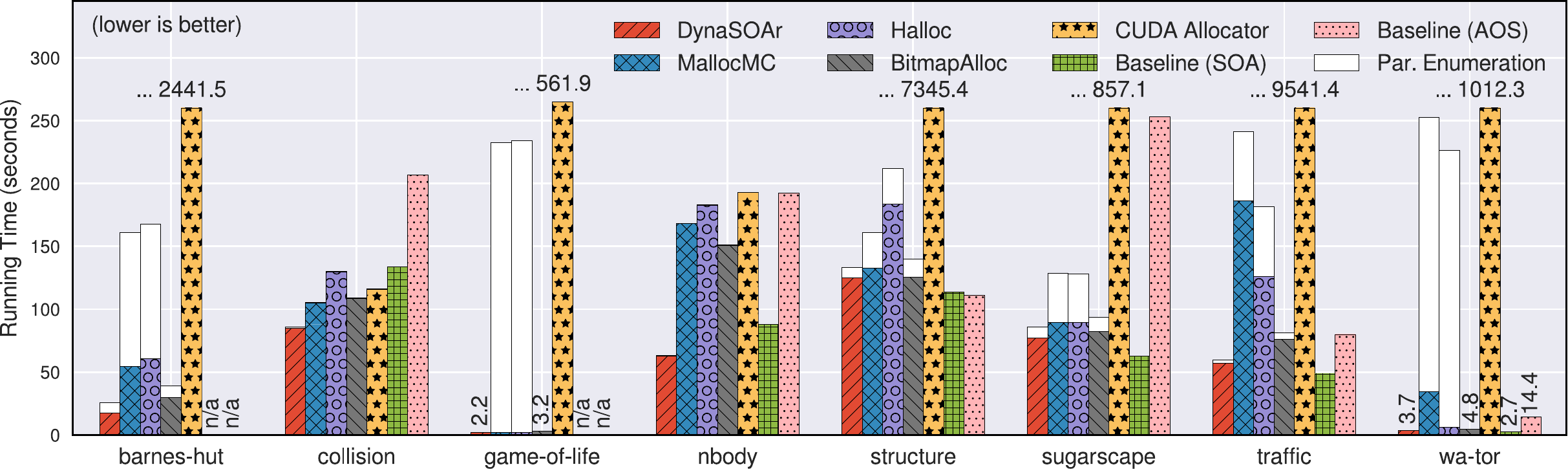}
  \caption{Running Time of SMMO Application Benchmarks. We gave every allocator some extra memory to avoid memory scarcity slowdowns: The heap size is 8~GiB, at least 4 times bigger than the maximum amount of all allocated memory at any point throughout the program execution.}
  \label{fig:bench_overview}
  \vspace{20pt}

  \includegraphics[width=\textwidth]{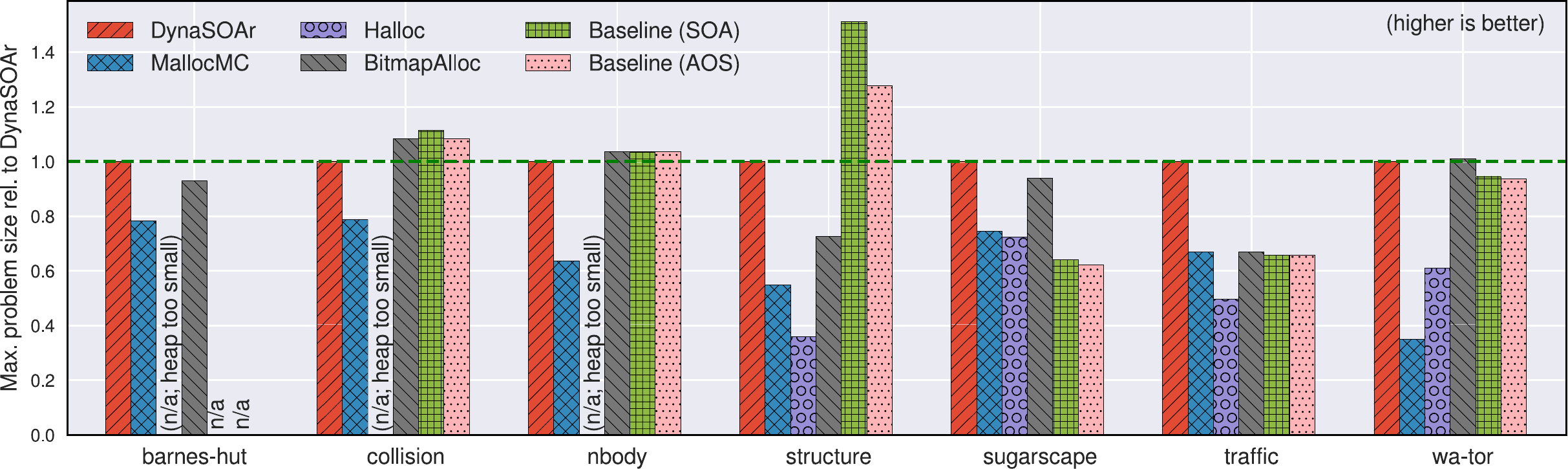}
  \caption{Space Efficiency.  We measured the max. problem size of every allocator with the same heap size. Does not take into account enumeration arrays. Results are relative to \textsc{DynaSOAr}.}
  \label{fig:memusage}
  \vspace{20pt}

  \centering
  \begin{subfigure}[t]{0.485\textwidth}
    \centering
    \includegraphics[width=\textwidth]{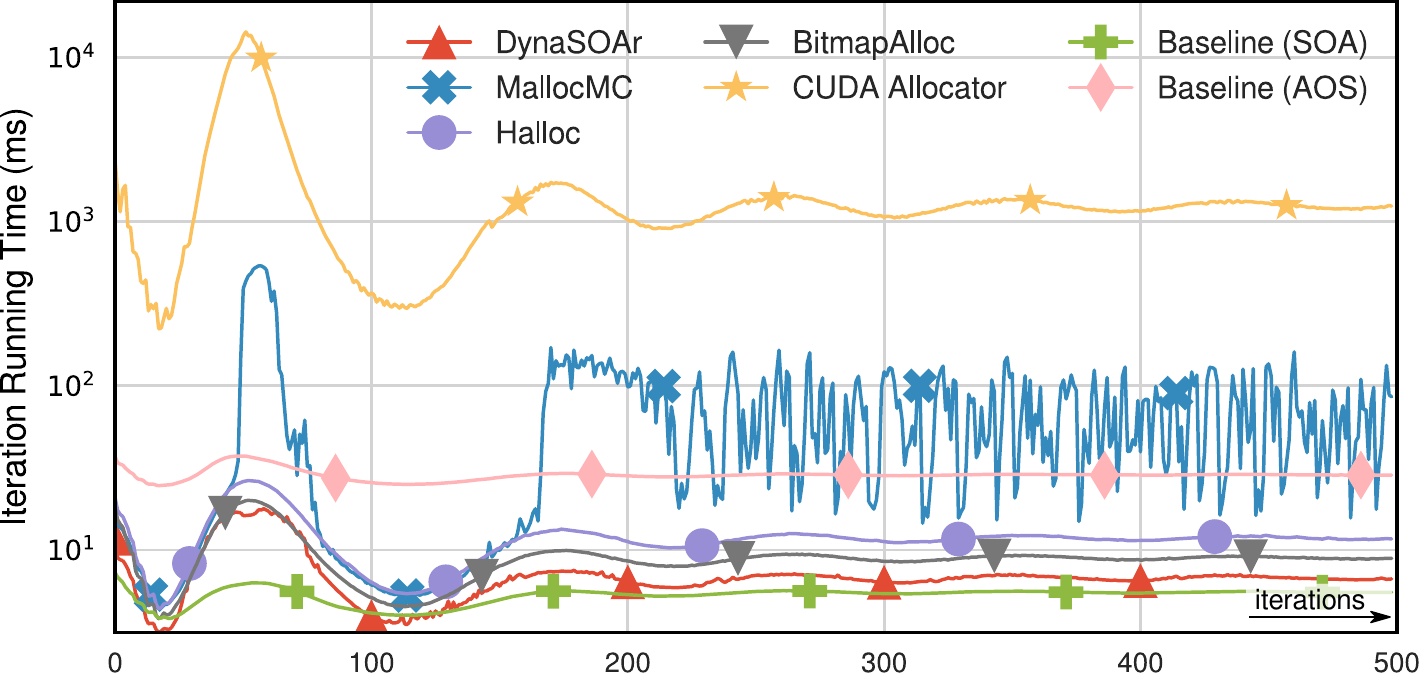}
    \caption*{\footnotesize \textbf{\textsf{(a)}} Comparison with other allocators.}  
  \end{subfigure}\hfill
  \begin{subfigure}[t]{0.485\textwidth}
    \centering
    \includegraphics[width=\textwidth]{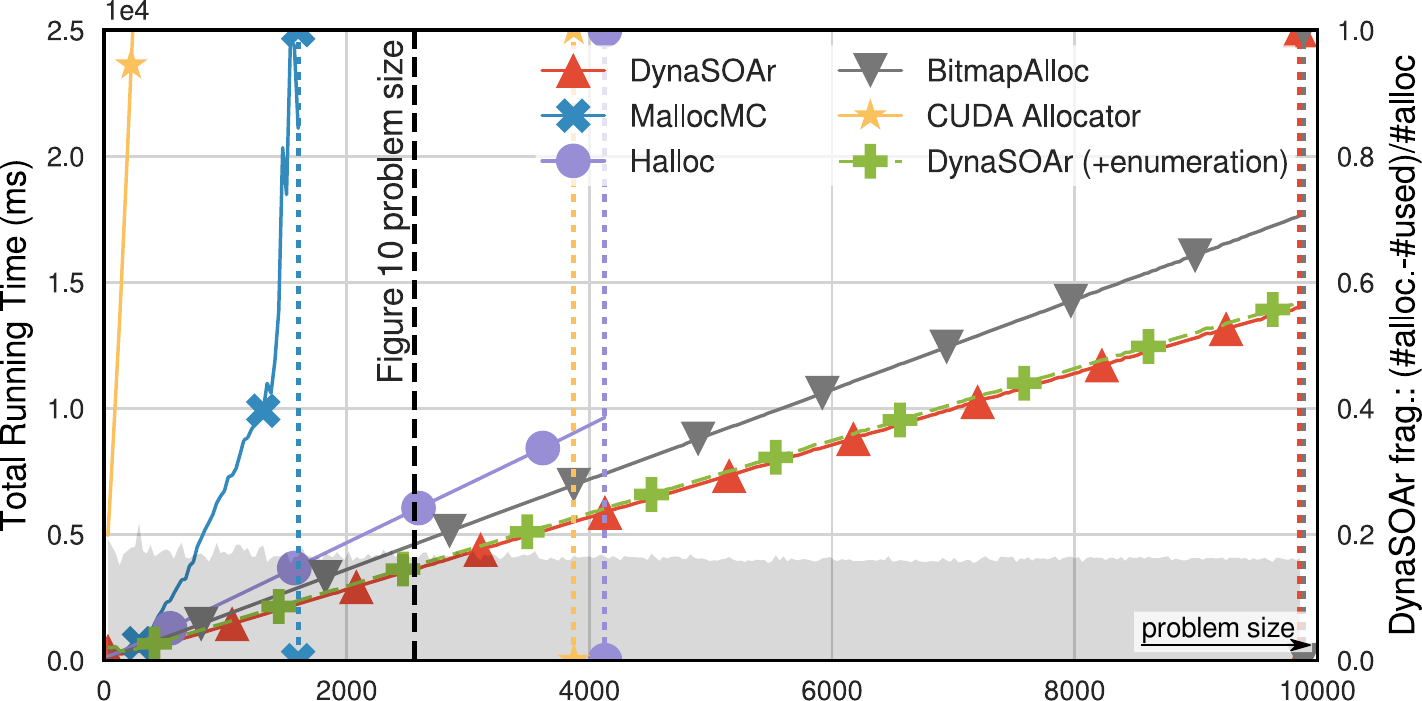}
    \caption*{\footnotesize \textbf{\textsf{(b)}} Fixed heap size, increasing problem size.} 
  \end{subfigure}

  \vspace{0.4cm}

  \begin{subfigure}[t]{0.485\textwidth}
    \centering
    \includegraphics[width=\textwidth]{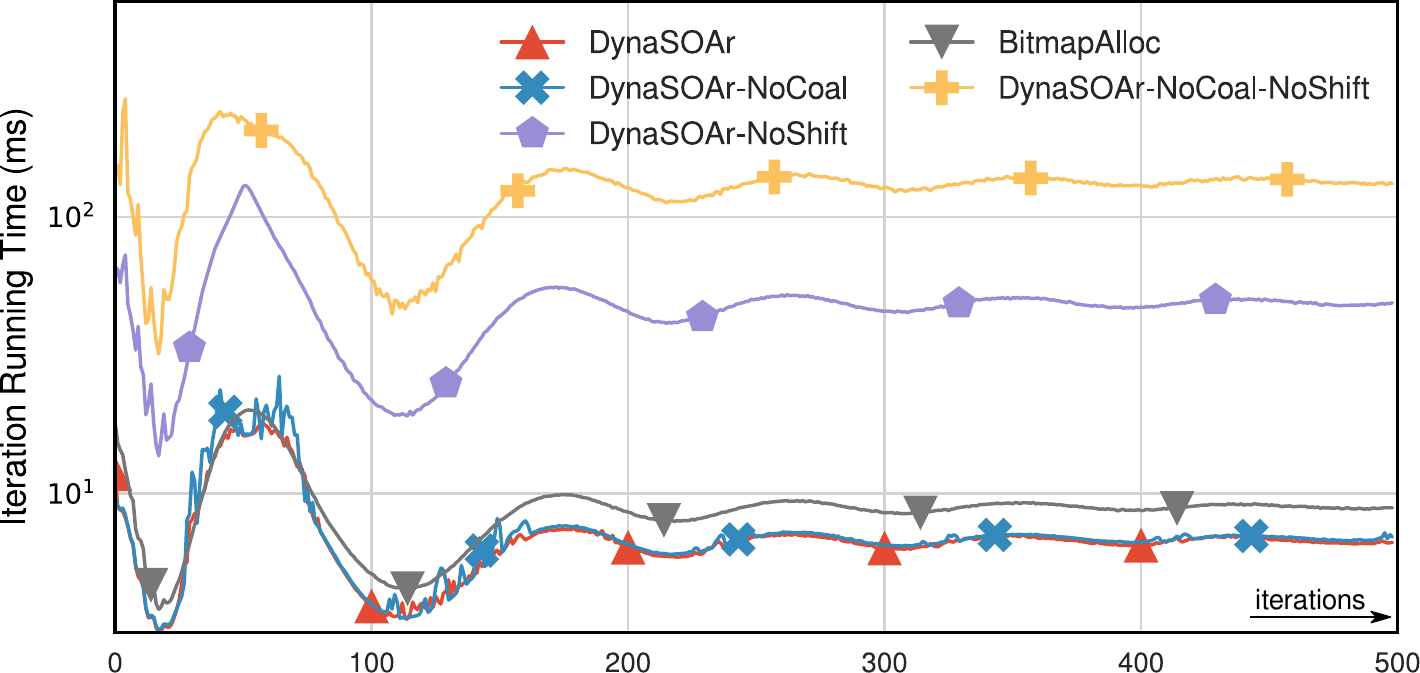}
    \caption*{\footnotesize \textbf{\textsf{(c)}} Isolating single \soaalloc{} optimizations.} 
  \end{subfigure}\hfill
  \begin{subfigure}[t]{0.485\textwidth}
    \centering
    \includegraphics[width=\textwidth]{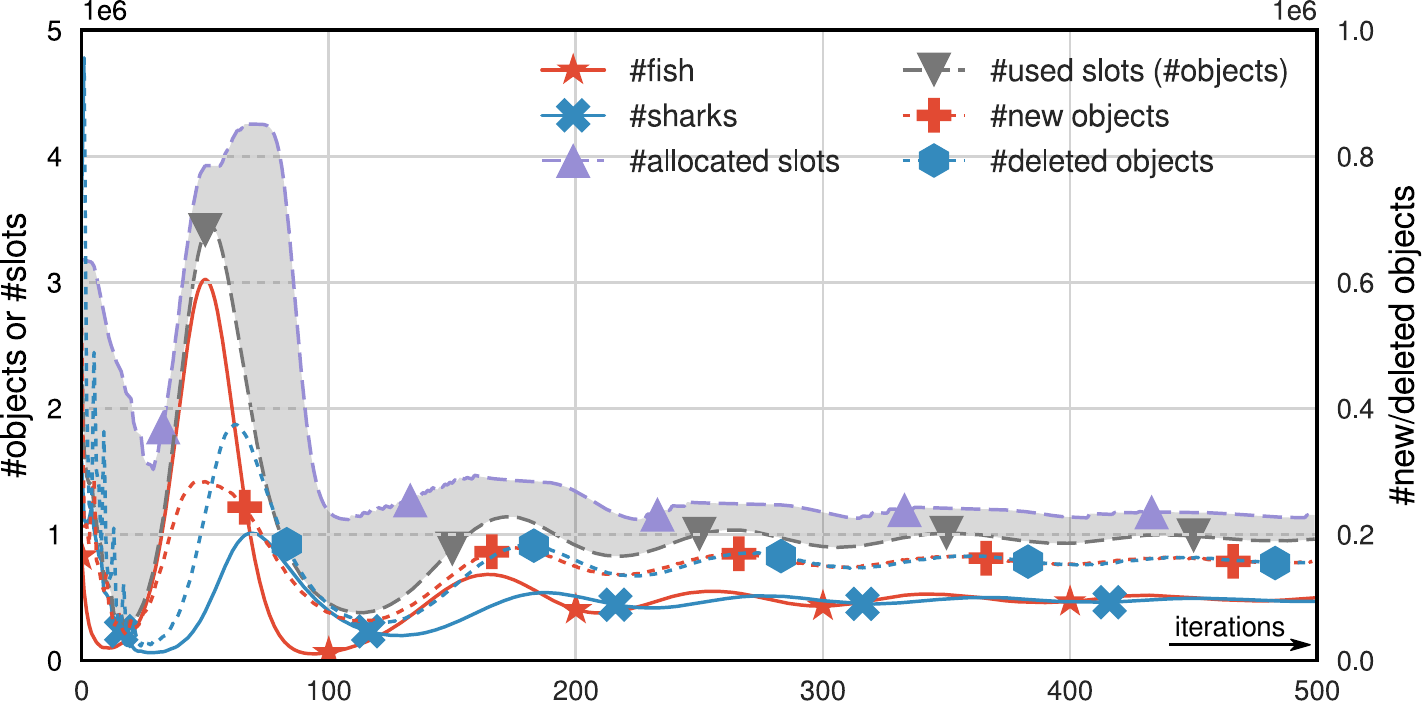}
    \caption*{\footnotesize \textbf{\textsf{(d)}} Number of (de)allocations and fragmentation.}   
  \end{subfigure}
  \caption{Detailed Analysis of \textsf{wa-tor}. Does not include enumeration time, unless indicated.}\label{fig:detailed_wa_tor}

  \label{fig:benchmark_scaling}
\end{figure}


\subsection{Space Efficiency}
\label{sec:bench_mem_usage_sec}
To evaluate how efficiently allocators manage memory, we gave them the same heap size and experimentally determined the max. problem size before running out of memory (Fig.~\ref{fig:memusage}).


For category (1) and (2) applications that allocate all memory during startup (\textsf{collision}, \textsf{nbody}, \textsf{structure}), the baseline versions are more space-efficient. The exact number of objects per type is known ahead of time, so placing objects in memory is trivial. However, even though category (2) applications delete objects throughout their runtime, the memory consumption of the baseline versions does not decrease over time. This is a problem even for \textsc{DynaSOAr} because blocks can only be deleted when they are entirely empty, which can take some time.

Category (3) applications (\textsf{sugarscape}, \textsf{traffic}, \textsf{wa-tor}) exhibit a fixed grid/network structure of cells, upon which a dynamic set of agents is moving. The baseline versions allocate the fields of agents directly inside cells. Classes for agents are combined with the cell class and some fields have \texttt{null} values (or garbage) if they are not used. This wastes memory because not all cells are occupied by agents all the time. Here, \soaalloc{} is not as fast as optimized SOA baseline implementations, but it can handle significantly larger problem sizes.

Out of all allocators, \soaalloc{} is most space-efficient. MallocMC and Halloc are based on a hashing approach. With rising heap fill levels, it becomes increasingly difficult to find free memory for allocations, so they fail to use the entire heap memory. \soaalloc{} and BitmapAlloc can avoid this problem with bitmaps, which act as an index for free memory.

Albeit negligible in these benchmarks, \soaalloc{} and Baseline (SOA) also benefit from slightly smaller object sizes: Only SOA arrays must be aligned and not every object.

\subsection{Detailed Analysis of \textsf{wa-tor}}
\label{sec:detailed_analysis_wator_s}
\textsf{wa-tor} is a particularly interesting benchmark. It exhibits a massive number of (de)allocations in waves, until an equilibrium between fish and sharks is reached. This allows us to measure performance at a massive and at a lower number of concurrent (de)allocations. For a fair comparison of allocators, we do not include time spent on enumeration in this section.

Fig.~\ref{fig:detailed_wa_tor}\textsc{a} shows that \soaalloc{} always provides superior performance compared to other allocators; during (de)allocation spikes (around iteration 50), as well as if fewer concurrent (de)allocations take place. The performance of mallocMC degrades after a few iterations and does not recover, possibly due to a fragmented heap.

In (\textsc{b}), all allocators were given a heap size of 1~GB and the problem size increases gradually on the x-axis. mallocMC performs well at first, but its performance drops rapidly as soon as the heap starts filling up. \soaalloc{} can handle much larger problem sizes, given the same amount of heap memory. The running time grows linearly with the problem size, showing that recent GPU architectures can handle atomic operations quite well.

Fragmentation in \soaalloc{} is different from other allocators: \soaalloc{} does not have internal or external fragmentation by design, but memory within allocated blocks is only available for a certain type. This sort of fragmentation decreases with better clustering. In \textsc{DynaSOAr}, fragmentation $F$ is the relative number of unused objects slots among all allocated blocks $\mathit{Blocks}$ (gray area in (\textsc{b}) and (\textsc{d})).


\begin{align*}
F = \frac{\sum_{b \in \mathit{Blocks}} (N_{\mathit{type}(b)} - \mathit{used}(b))}{\sum_{b \in \mathit{Blocks}}  N_{\mathit{type}(b)}}  \approx \frac{1}{\mbox{\#blocks}} \sum_{b \in \mathit{Blocks}} \frac{\mbox{\#free slots}(b)}{\mbox{\#slots}(b)} \tag{\emph{fragmentation}}
\end{align*}

At iterations~60--80 in (\textsc{d}), \soaalloc{} has high fragmentation because many fish objects were deallocated. However, a block can only be deallocated when \emph{all} of its objects are deallocated. The fragmentation level decreases gradually because new allocations are performed in existing (active) blocks. Therefore, new blocks are rarely allocated and there is a chance that an active block will eventually run empty. As can be seen in (\textsc{b}), fragmentation is independent of the problem size and constant at around 18\% after 500 Wa-Tor iterations. 


\begin{figure}
  \begin{minipage}[c]{0.485\textwidth}
    \includegraphics[width=\textwidth]{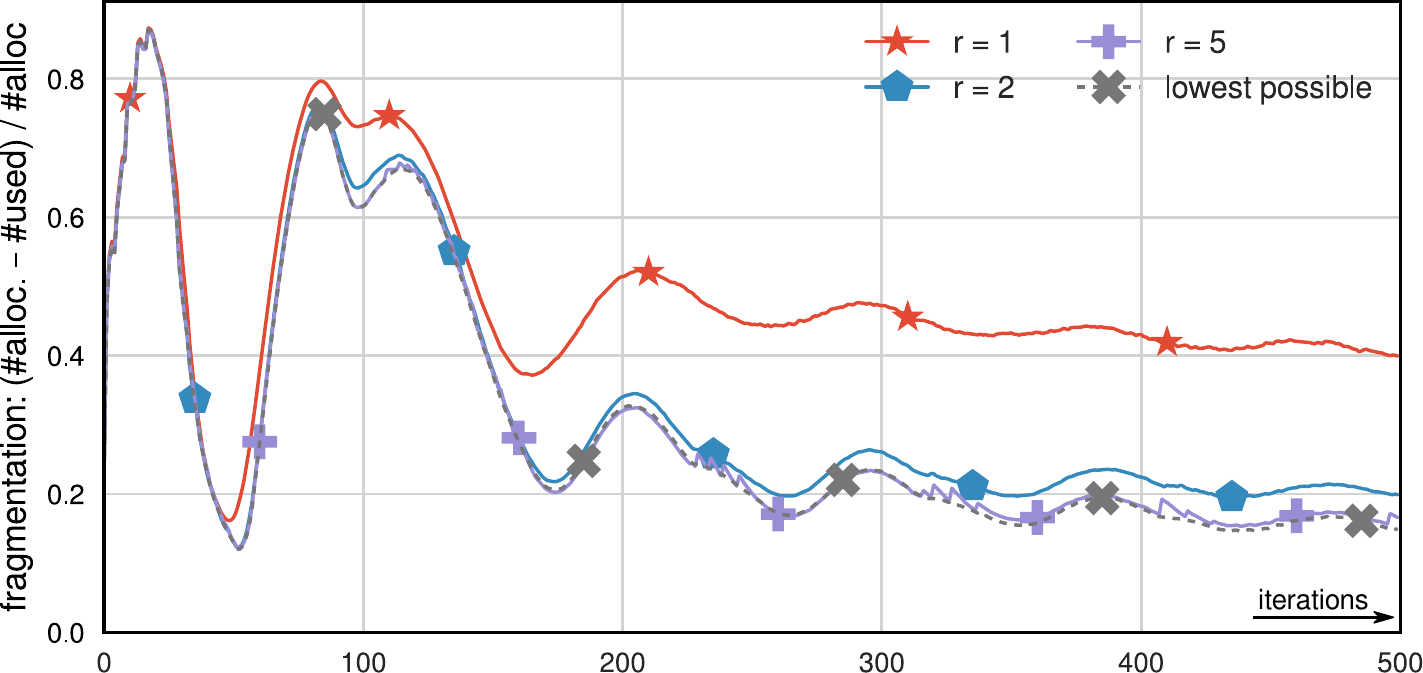}
  \end{minipage}\hfill
  \begin{minipage}[c]{0.475\textwidth}
    \caption{Memory fragmentation (\textsf{wa-tor}) by \#active block lookup attemps $r$ (Alg.~\ref{alg:block_allocate}, Line~2). With only 1 retry ($r=2$), frag. is reduced by 50\%. \textsc{DynaSOAr} uses $r=5$ by default, which is close to the lowest achievable frag. level (i.e., without thread contention). Due to unfortunate alloc.-delloc. patterns, a frag. rate of 0\% is not achievable without manually relocating objects or predicting future (de)allocations.} \label{fig:frag_param}
  \end{minipage}
\end{figure}

We implemented multiple \soaalloc{} variants to pinpoint the source of \soaalloc{}'s speedup over other allocators (Fig.~\ref{fig:detailed_wa_tor}\textsc{c}. The most important optimization is the rotation-shifting of bitmaps. Without shifting (\textsf{*-NoShift}), performance degrades severely due to thread contention. Allocation request coalescing is another optimization that reduces thread contention significantly (compare \textsf{DynaSOAr-NoCoal-NoShift} and \textsf{DynaSOAr-NoShift}), but it cannot improve performance much further if we are already rotation-shifting bitmaps (compare \textsf{DynaSOAr} and \textsf{DynaSOAr-NoCoal}).

In Fig.~\ref{fig:frag_param}, we experiment with the number of active block lookup attempts before entering the slow path, which strongly affects fragmentation. 


\subsection{Raw Allocation Performance}
The \emph{Linux Scalability} microbenchmark~\cite{Lever:2000:MPM:1267724.1267780} measures the raw (de)allocation time of allocators. We set the heap size to 1~GiB and one CUDA kernel allocates $n$ 64-byte objects in each of the 16,384 threads. A second CUDA kernel deallocates all objects. Allocated memory is never accessed. In Fig.~\ref{fig:benchmark_scaling_2}\textsc{a}, the x-axis denotes the number of allocations per thread $n$ and the y-axis shows the total benchmark running time divided by $n$.

We chose the size of the heap such that it can hold exactly $16384 \times n$ objects with $n=1024$ (100\% heap utilization). No allocator can reach perfect utilization because some memory is used for internal data structures such as bitmaps.


Halloc is the fastest allocator. Both Halloc and mallocMC fail to allocate more than 510 objects (49.8\% utilization). This is better than in some other benchmarks, likely because only objects of one size are allocated. \soaalloc{} (96.9\% utilization), BitmapAlloc (98.4\% utilization) and Halloc scale almost perfectly with the number of allocations.

\begin{figure}
  \centering
  \begin{subfigure}[t]{0.485\textwidth}
    \centering
    \includegraphics[width=\textwidth]{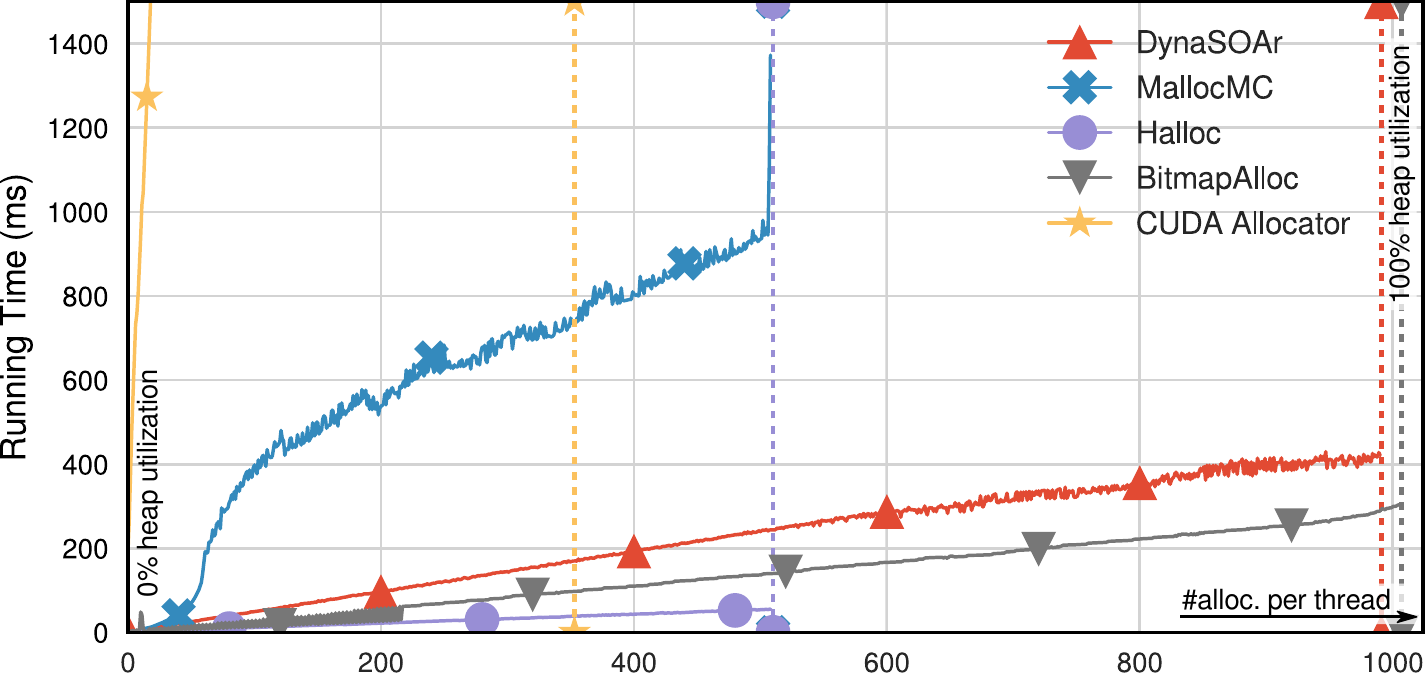}
    \caption*{\footnotesize \textbf{\textsf{(a)}} Linux Scalability: Increasing \#allocations.}   
  \end{subfigure}\hfill
  \begin{subfigure}[t]{0.485\textwidth}
    \centering
    \includegraphics[width=\textwidth]{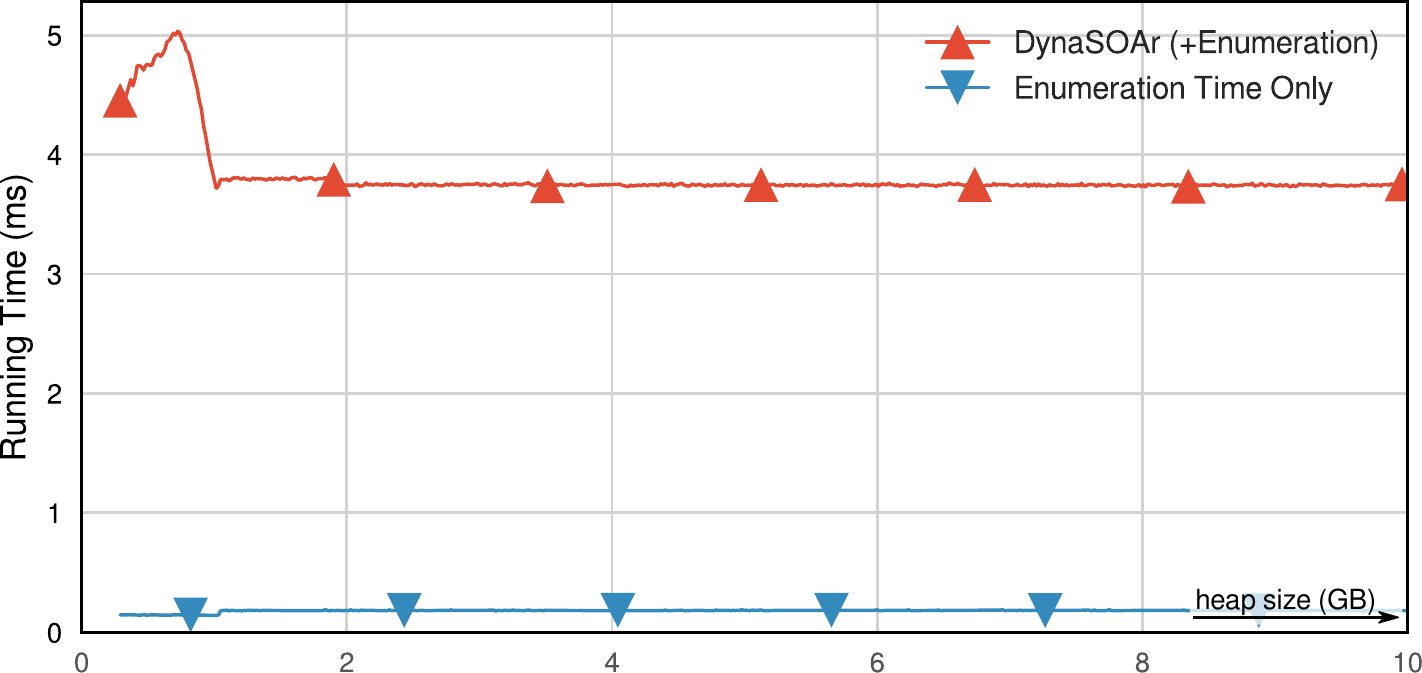}
    \caption*{\footnotesize \textbf{\textsf{(b)}} Scaling study: Heap size (\textsf{wa-tor}).} 
  \end{subfigure}
  \caption{Scaling Study: Number of Allocations and Heap Size.}
  \label{fig:benchmark_scaling_2}
\end{figure}




\subsection{Parallel Object Enumeration}
The overhead of object enumeration (parallel do-all) is negible in most benchmarks (Fig.~\ref{fig:bench_overview}, Fig.~\ref{fig:benchmark_scaling}\textsc{b}). In Fig.~\ref{fig:benchmark_scaling_2}\textsc{b}, the problem size is fixed but the heap size increases on the x-axis. \textsc{DynaSOAr}'s performance (and that of object enumeration) is independent of the size of the heap, if enough memory is available for the application. This shows that our hierarchical bitmaps work well with various heap sizes.



\section{Conclusion}
\label{sec:conclusion}
We presented \soaalloc{}, a new dynamic object allocator for SIMD architectures. The main insight of our work is that memory allocators should not only aim for good raw (de)allocation performance, but also optimize the usage of allocated memory. \soaalloc{} was designed for GPUs, but its basic ideas are applicable to other architectures and systems with good or guaranteed vectorization such as the Intel SPMD compiler~\cite{6339601}.

\soaalloc{} achieves good memory access performance by controlling (a) memory allocation and (b) memory access with a parallel do-all operation. \textsc{DynaSOAr}'s main speedup over other allocators is due to an SOA-style object layout, which can benefit memory bandwidth utilization (through coalesced memory access) and cache utilization. To allow for dynamic (de)allocation of objects, \textsc{DynaSOAr} allocates objects in blocks instead of a plain SOA layout. \textsc{DynaSOAr} utilizes hierarchical bitmaps for fast and compact allocations with low fragmentation.

Our benchchmarks show that \soaalloc{} can achieve significant speedups over state-of-the-art allocators of more than 3x in application code with structured data, due to better memory access performance. \soaalloc{} also has a significantly lower memory footprint than other allocators, mainly because \textsc{DynaSOAr} has no internal fragmentation by design and is not based on hashing. Our work also shows how an SOA layout can support class inheritance without wasting memory: by allocating objects in blocks and encoding block sizes in object pointers.

In the future, we will investigate how \textsc{DynaSOAr} can be extended to support virtual functions and other custom object layouts.




\bibliography{paper}

\begin{thebibliography}{10}

\bibitem{Abel:1999:ATS}
James Abel, Kumar Balasubramanian, Mike Bargeron, Tom Craver, and Mike Phlipot.
\newblock Applications tuning for streaming {SIMD} extensions.
\newblock {\em Intel Technology Journal}, (Q2):13, May 1999.

\bibitem{warp_aggre}
Andy Adinets.
\newblock {CUDA} pro tip: Optimized filtering with warp-aggregated atomics.
\newblock
  \href{https://devblogs.nvidia.com/cuda-pro-tip-optimized-filtering-warp-aggregated-atomics/}{\texttt{https://devblogs.nvidia.com/cuda-pro-tip-optimized-filtering-warp-
  aggregated-atomics/}}, 2017.

\bibitem{hallocweb}
Andrew~V. Adinetz and Dirk Pleiter.
\newblock Halloc: A high-throughput dynamic memory allocator for {GPGPU}
  architectures.
\newblock In {\em GPU Technology Conference 2014}, 2014.

\bibitem{ALEXANDER1998113}
Stephen~G. Alexander and Craig~B. Agnor.
\newblock N-body simulations of late stage planetary formation with a simple
  fragmentation model.
\newblock {\em Icarus}, 132(1):113--124, 1998.
\newblock \href {http://dx.doi.org/10.1006/icar.1998.5905}
  {\path{doi:10.1006/icar.1998.5905}}.

\bibitem{Alglave:2015:GCW:2694344.2694391}
Jade Alglave, Mark Batty, Alastair~F. Donaldson, Ganesh Gopalakrishnan, Jeroen
  Ketema, Daniel Poetzl, Tyler Sorensen, and John Wickerson.
\newblock {GPU} concurrency: Weak behaviours and programming assumptions.
\newblock In {\em Proceedings of the Twentieth International Conference on
  Architectural Support for Programming Languages and Operating Systems},
  ASPLOS '15, pages 577--591, New York, NY, USA, 2015. ACM.
\newblock \href {http://dx.doi.org/10.1145/2694344.2694391}
  {\path{doi:10.1145/2694344.2694391}}.

\bibitem{allan2010survey}
Robert~J. Allan.
\newblock Survey of agent based modelling and simulation tools.
\newblock Technical Report DL-TR-2010-007, Science and Technology Facilities
  Council, Warrington, United Kingdom, Oct 2010.

\bibitem{DBLP:journals/corr/abs-1710-11246}
Saman Ashkiani, Martin Farach{-}Colton, and John~D. Owens.
\newblock A dynamic hash table for the {GPU}.
\newblock {\em CoRR}, abs/1710.11246, 2017.

\bibitem{IJNC151}
Darius Bakunas-Milanowski, Vernon Rego, Janche Sang, and Chansu Yu.
\newblock Efficient algorithms for stream compaction on {GPUs}.
\newblock {\em International Journal of Networking and Computing},
  7(2):208--226, 2017.
\newblock \href {http://dx.doi.org/10.15803/ijnc.7.2_208}
  {\path{doi:10.15803/ijnc.7.2_208}}.

\bibitem{bandini2009}
Stefania Bandini, Sara Manzoni, and Giuseppe Vizzari.
\newblock Agent based modeling and simulation: An informatics perspective.
\newblock {\em Journal of Artificial Societies and Social Simulation}, 12(4):4,
  2009.

\bibitem{cpp_placement}
Eli Bendersky.
\newblock The many faces of operator new in {C++}.
\newblock
  \url{https://eli.thegreenplace.net/2011/02/17/the-many-faces-of-operator-new-in-c},
  2011.

\bibitem{Berger:2000:HSM:378993.379232}
Emery~D. Berger, Kathryn~S. McKinley, Robert~D. Blumofe, and Paul~R. Wilson.
\newblock Hoard: A scalable memory allocator for multithreaded applications.
\newblock In {\em Proceedings of the Ninth International Conference on
  Architectural Support for Programming Languages and Operating Systems},
  ASPLOS IX, pages 117--128, New York, NY, USA, 2000. ACM.
\newblock \href {http://dx.doi.org/10.1145/378993.379232}
  {\path{doi:10.1145/378993.379232}}.

\bibitem{intel_aos_soa}
Paul Besl.
\newblock A case study comparing {AoS} (arrays of structures) and {SoA}
  (structures of arrays) data layouts for a compute-intensive loop run on
  {Intel Xeon} processors and {Intel Xeon Phi} product family coprocessors.
\newblock Technical report, Intel Corporation, 2013.

\bibitem{Billeter:2009:ESC:1572769.1572795}
Markus Billeter, Ola Olsson, and Ulf Assarsson.
\newblock Efficient stream compaction on wide {SIMD} many-core architectures.
\newblock In {\em Proceedings of the Conference on High Performance Graphics
  2009}, HPG '09, pages 159--166, New York, NY, USA, 2009. ACM.
\newblock \href {http://dx.doi.org/10.1145/1572769.1572795}
  {\path{doi:10.1145/1572769.1572795}}.

\bibitem{Blumofe:1999:SMC:324133.324234}
Robert~D. Blumofe and Charles~E. Leiserson.
\newblock Scheduling multithreaded computations by work stealing.
\newblock {\em J. ACM}, 46(5):720--748, September 1999.
\newblock \href {http://dx.doi.org/10.1145/324133.324234}
  {\path{doi:10.1145/324133.324234}}.

\bibitem{Brown:2015:RML:2767386.2767436}
Trevor~Alexander Brown.
\newblock Reclaiming memory for lock-free data structures: There has to be a
  better way.
\newblock In {\em Proceedings of the 2015 ACM Symposium on Principles of
  Distributed Computing}, PODC '15, pages 261--270, New York, NY, USA, 2015.
  ACM.
\newblock \href {http://dx.doi.org/10.1145/2767386.2767436}
  {\path{doi:10.1145/2767386.2767436}}.

\bibitem{BURTSCHER201175}
Martin Burtscher and Keshav Pingali.
\newblock Chapter 6 -- an efficient {CUDA} implementation of the tree-based
  {Barnes Hut} n-body algorithm.
\newblock In Wen mei W.~Hwu, editor, {\em GPU Computing Gems Emerald Edition},
  Applications of GPU Computing Series, pages 75 -- 92. Morgan Kaufmann,
  Boston, 2011.
\newblock \href {http://dx.doi.org/10.1016/B978-0-12-384988-5.00006-1}
  {\path{doi:10.1016/B978-0-12-384988-5.00006-1}}.

\bibitem{CARY199720}
John~R. Cary, Svetlana~G. Shasharina, Julian~C. Cummings, John~V.W. Reynders,
  and Paul~J. Hinker.
\newblock Comparison of {C++} and {Fortran 90} for object-oriented scientific
  programming.
\newblock {\em Computer Physics Communications}, 105(1):20--36, 1997.
\newblock \href {http://dx.doi.org/10.1016/S0010-4655(97)00043-X}
  {\path{doi:10.1016/S0010-4655(97)00043-X}}.

\bibitem{Chilimbi:1999:CSD:301618.301635}
Trishul~M. Chilimbi, Bob Davidson, and James~R. Larus.
\newblock Cache-conscious structure definition.
\newblock In {\em Proceedings of the ACM SIGPLAN 1999 Conference on Programming
  Language Design and Implementation}, PLDI '99, pages 13--24, New York, NY,
  USA, 1999. ACM.
\newblock \href {http://dx.doi.org/10.1145/301618.301635}
  {\path{doi:10.1145/301618.301635}}.

\bibitem{nvidia_memoryco}
NVIDIA Corporation.
\newblock {CUDA C} best practices guide.
\newblock
  \url{https://docs.nvidia.com/cuda/cuda-c-best-practices-guide/index.html\#coalesced-access-to-global-memory},
  2018.

\bibitem{doi:10.1002/9781119332015.ch3}
Cederman Daniel, Gidenstam Anders, Ha~Phuong, Sundell Hkan, Papatriantafilou
  Marina, and Tsigas Philippas.
\newblock {\em Lock-Free Concurrent Data Structures}, chapter~3, pages 59--79.
\newblock Wiley-Blackwell, 2017.
\newblock \href {http://dx.doi.org/10.1002/9781119332015.ch3}
  {\path{doi:10.1002/9781119332015.ch3}}.

\bibitem{10.1007/978-3-540-25934-3_2}
Kei Davis and J{\"o}rg Striegnitz.
\newblock Parallel object-oriented scientific computing today.
\newblock In Frank Buschmann, Alejandro~P. Buchmann, and Mariano~A. Cilia,
  editors, {\em Object-Oriented Technology. ECOOP 2003 Workshop Reader}, pages
  11--16, Berlin, Heidelberg, 2004. Springer-Verlag.
\newblock \href {http://dx.doi.org/10.1007/978-3-540-25934-3_2}
  {\path{doi:10.1007/978-3-540-25934-3_2}}.

\bibitem{DeGonzalo:2019:AGW:3314872.3314884}
Simon~Garcia De~Gonzalo, Sitao Huang, Juan G\'{o}mez-Luna, Simon Hammond, Onur
  Mutlu, and Wen-mei Hwu.
\newblock Automatic generation of warp-level primitives and atomic instructions
  for fast and portable parallel reduction on {GPUs}.
\newblock In {\em Proceedings of the 2019 IEEE/ACM International Symposium on
  Code Generation and Optimization}, CGO 2019, pages 73--84, Piscataway, NJ,
  USA, Feb 2019. IEEE Press.
\newblock \href {http://dx.doi.org/10.1109/CGO.2019.8661187}
  {\path{doi:10.1109/CGO.2019.8661187}}.

\bibitem{10.2307/24969495}
Alexander~K. Dewdney.
\newblock Computer creations: Sharks and fish wage an ecological war on the
  toroidal planet {Wa-Tor}.
\newblock {\em Scientific American}, 251(6):14--26, 12 1984.

\bibitem{eckert_carlchristian_helmut_johannes_2014_34461}
Carlchristian H.~J. Eckert.
\newblock {Enhancements of the massively parallel memory allocator ScatterAlloc
  and its adaption to the general interface mallocMC}.
\newblock Junior thesis. Technische Universit{\"a}t Dresden, October 2014.
\newblock \href {http://dx.doi.org/10.5281/zenodo.34461}
  {\path{doi:10.5281/zenodo.34461}}.

\bibitem{osti_1398234}
Harold~C. Edwards and Daniel~A. Ibanez.
\newblock Kokkos' task {DAG} capabilities.
\newblock Technical Report SAND2017-10464, Sandia National Laboratories,
  Albuquerque, New Mexico, USA, 9 2017.
\newblock \href {http://dx.doi.org/10.2172/1398234}
  {\path{doi:10.2172/1398234}}.

\bibitem{Ellen:2007:SSN:1281100.1281106}
Faith Ellen, Yossi Lev, Victor Luchangco, and Mark Moir.
\newblock {SNZI}: Scalable nonzero indicators.
\newblock In {\em Proceedings of the Twenty-sixth Annual ACM Symposium on
  Principles of Distributed Computing}, PODC '07, pages 13--22, New York, NY,
  USA, 2007. ACM.
\newblock \href {http://dx.doi.org/10.1145/1281100.1281106}
  {\path{doi:10.1145/1281100.1281106}}.

\bibitem{RePEc:mtp:titles:0262550253}
Joshua~M. Epstein and Robert Axtell.
\newblock {\em Growing Artificial Societies: Social Science from the Bottom
  Up}, volume~1.
\newblock The MIT Press, 1 edition, 1996.

\bibitem{FORDE1990355}
Bruce~W.R. Forde, Ricardo~O. Foschi, and Siegfried~F. Stiemer.
\newblock Object-oriented finite element analysis.
\newblock {\em Computers \& Structures}, 34(3):355--374, 1990.
\newblock \href {http://dx.doi.org/10.1016/0045-7949(90)90261-Y}
  {\path{doi:10.1016/0045-7949(90)90261-Y}}.

\bibitem{Franco:2017:YAG:3133850.3133861}
Juliana Franco, Martin Hagelin, Tobias Wrigstad, Sophia Drossopoulou, and Susan
  Eisenbach.
\newblock You can have it all: Abstraction and good cache performance.
\newblock In {\em Proceedings of the 2017 ACM SIGPLAN International Symposium
  on New Ideas, New Paradigms, and Reflections on Programming and Software},
  Onward! 2017, pages 148--167, New York, NY, USA, 2017. ACM.
\newblock \href {http://dx.doi.org/10.1145/3133850.3133861}
  {\path{doi:10.1145/3133850.3133861}}.

\bibitem{97c20025105249ca9f87087e9d7ec2c8}
Dietma Gallistl.
\newblock The adaptive finite element method.
\newblock {\em Snapshots of modern mathematics from Oberwolfach}, 13, 2016.
\newblock \href {http://dx.doi.org/10.14760/SNAP-2016-013-EN}
  {\path{doi:10.14760/SNAP-2016-013-EN}}.

\bibitem{Gelado:2019:TGM:3293883.3295727}
Isaac Gelado and Michael Garland.
\newblock Throughput-oriented {GPU} memory allocation.
\newblock In {\em Proceedings of the 24th Symposium on Principles and Practice
  of Parallel Programming}, PPoPP '19, pages 27--37, New York, NY, USA, 2019.
  ACM.
\newblock \href {http://dx.doi.org/10.1145/3293883.3295727}
  {\path{doi:10.1145/3293883.3295727}}.

\bibitem{Grunwald:1993:ICL:155090.155107}
Dirk Grunwald, Benjamin Zorn, and Robert Henderson.
\newblock Improving the cache locality of memory allocation.
\newblock In {\em Proceedings of the ACM SIGPLAN 1993 Conference on Programming
  Language Design and Implementation}, PLDI '93, pages 177--186, New York, NY,
  USA, 1993. ACM.
\newblock \href {http://dx.doi.org/10.1145/155090.155107}
  {\path{doi:10.1145/155090.155107}}.

\bibitem{Harish:2007:ALG:1782174.1782200}
Pawan Harish and P.~J. Narayanan.
\newblock Accelerating large graph algorithms on the {GPU} using {CUDA}.
\newblock In {\em Proceedings of the 14th International Conference on High
  Performance Computing}, HiPC'07, pages 197--208, Berlin, Heidelberg, 2007.
  Springer-Verlag.
\newblock \href {http://dx.doi.org/10.1007/978-3-540-77220-0_21}
  {\path{doi:10.1007/978-3-540-77220-0_21}}.

\bibitem{cuda_grid_stride}
Mark Harris.
\newblock {CUDA} pro tip: Write flexible kernels with grid-stride loops.
\newblock
  \href{https://devblogs.nvidia.com/cuda-pro-tip-write-flexible-kernels-grid-stride-loops/}{\texttt{https://devblogs.nvidia.com/cuda-pro-tip-write-flexible-kernels-grid-stride
  -loops/}}, 2013.

\bibitem{cpp_obj}
Kevlin Henney.
\newblock Valued conversions.
\newblock {\em C++ Report}, 12:37--40, July 2000.

\bibitem{HOMANN2018325}
Holger Homann and Francois Laenen.
\newblock {SoAx}: A generic {C++} structure of arrays for handling particles in
  {HPC} codes.
\newblock {\em Computer Physics Communications}, 224:325--332, 2018.
\newblock \href {http://dx.doi.org/10.1016/j.cpc.2017.11.015}
  {\path{doi:10.1016/j.cpc.2017.11.015}}.

\bibitem{5577907}
Xiaohuang Huang, Christopher~I. Rodrigues, Stephen Jones, Ian Buck, and Wen-Mei
  Hwu.
\newblock Xmalloc: A scalable lock-free dynamic memory allocator for many-core
  machines.
\newblock In {\em 2010 10th IEEE International Conference on Computer and
  Information Technology}, pages 1134--1139, June 2010.
\newblock \href {http://dx.doi.org/10.1109/CIT.2010.206}
  {\path{doi:10.1109/CIT.2010.206}}.

\bibitem{5473222}
Byunghyun {Jang}, Dana {Schaa}, Perhaad {Mistry}, and David {Kaeli}.
\newblock Exploiting memory access patterns to improve memory performance in
  data-parallel architectures.
\newblock {\em IEEE Transactions on Parallel and Distributed Systems},
  22(1):105--118, January 2011.
\newblock \href {http://dx.doi.org/10.1109/TPDS.2010.107}
  {\path{doi:10.1109/TPDS.2010.107}}.

\bibitem{Kale:1993:CPC:165854.165874}
Laxmikant~V. Kale and Sanjeev Krishnan.
\newblock {CHARM++}: A portable concurrent object oriented system based on
  {C++}.
\newblock In {\em Proceedings of the Eighth Annual Conference on
  Object-oriented Programming Systems, Languages, and Applications}, OOPSLA
  '93, pages 91--108, New York, NY, USA, 1993. ACM.
\newblock \href {http://dx.doi.org/10.1145/165854.165874}
  {\path{doi:10.1145/165854.165874}}.

\bibitem{10.1007/978-3-662-48096-0_21}
Klaus Kofler, Biagio Cosenza, and Thomas Fahringer.
\newblock Automatic data layout optimizations for {GPUs}.
\newblock In Jesper~Larsson Tr{\"a}ff, Sascha Hunold, and Francesco Versaci,
  editors, {\em Euro-Par 2015: Parallel Processing}, pages 263--274, Berlin,
  Heidelberg, 2015. Springer-Verlag.
\newblock \href {http://dx.doi.org/10.1007/978-3-662-48096-0_21}
  {\path{doi:10.1007/978-3-662-48096-0_21}}.

\bibitem{7853809}
Florian Lemaitre and Lionel Lacassagne.
\newblock Batched cholesky factorization for tiny matrices.
\newblock In {\em 2016 Conference on Design and Architectures for Signal and
  Image Processing (DASIP)}, pages 130--137, Oct 2016.
\newblock \href {http://dx.doi.org/10.1109/DASIP.2016.7853809}
  {\path{doi:10.1109/DASIP.2016.7853809}}.

\bibitem{Lever:2000:MPM:1267724.1267780}
Chuck Lever and David Boreham.
\newblock Malloc() performance in a multithreaded {Linux} environment.
\newblock In {\em Proceedings of the Annual Conference on USENIX Annual
  Technical Conference}, ATEC '00, Berkeley, CA, USA, 2000. USENIX Association.

\bibitem{Li:2014:ENS:2701002.2701020}
Xiaosong Li, Wentong Cai, and Stephen~J. Turner.
\newblock Efficient neighbor searching for agent-based simulation on {GPU}.
\newblock In {\em Proceedings of the 2014 IEEE/ACM 18th International Symposium
  on Distributed Simulation and Real Time Applications}, DS-RT '14, pages
  87--96, Washington, DC, USA, 2014. IEEE Computer Society.
\newblock \href {http://dx.doi.org/10.1109/DS-RT.2014.19}
  {\path{doi:10.1109/DS-RT.2014.19}}.

\bibitem{Li:2015:CAS:2769458.2769470}
Xiaosong Li, Wentong Cai, and Stephen~J. Turner.
\newblock Cloning agent-based simulation on {GPU}.
\newblock In {\em Proceedings of the 3rd ACM SIGSIM Conference on Principles of
  Advanced Discrete Simulation}, SIGSIM PADS '15, pages 173--182, New York, NY,
  USA, 2015. ACM.
\newblock \href {http://dx.doi.org/10.1145/2769458.2769470}
  {\path{doi:10.1145/2769458.2769470}}.

\bibitem{doi:10.1002/cpe.3808}
Xiaosong Li, Wentong Cai, and Stephen~J. Turner.
\newblock Supporting efficient execution of continuous space agent-based
  simulation on {GPU}.
\newblock {\em Concurrency and Computation: Practice and Experience},
  28(12):3313--3332, 2016.
\newblock \href {http://dx.doi.org/10.1002/cpe.3808}
  {\path{doi:10.1002/cpe.3808}}.

\bibitem{LU2018240}
X.~Lu, B.Y. Chen, V.B.C. Tan, and T.E. Tay.
\newblock Adaptive floating node method for modelling cohesive fracture of
  composite materials.
\newblock {\em Engineering Fracture Mechanics}, 194:240--261, 2018.
\newblock \href {http://dx.doi.org/10.1016/j.engfracmech.2018.03.011}
  {\path{doi:10.1016/j.engfracmech.2018.03.011}}.

\bibitem{Mattis:2015:COI:2814228.2814230}
Toni Mattis, Johannes Henning, Patrick Rein, Robert Hirschfeld, and Malte
  Appeltauer.
\newblock Columnar objects: Improving the performance of analytical
  applications.
\newblock In {\em 2015 ACM International Symposium on New Ideas, New Paradigms,
  and Reflections on Programming and Software (Onward!)}, Onward! 2015, pages
  197--210, New York, NY, USA, 2015. ACM.
\newblock \href {http://dx.doi.org/10.1145/2814228.2814230}
  {\path{doi:10.1145/2814228.2814230}}.

\bibitem{Michael:2002:SMR:571825.571829}
Maged~M. Michael.
\newblock Safe memory reclamation for dynamic lock-free objects using atomic
  reads and writes.
\newblock In {\em Proceedings of the Twenty-first Annual Symposium on
  Principles of Distributed Computing}, PODC '02, pages 21--30, New York, NY,
  USA, 2002. ACM.
\newblock \href {http://dx.doi.org/10.1145/571825.571829}
  {\path{doi:10.1145/571825.571829}}.

\bibitem{Michael:2004:SLD:996841.996848}
Maged~M. Michael.
\newblock Scalable lock-free dynamic memory allocation.
\newblock In {\em Proceedings of the ACM SIGPLAN 2004 Conference on Programming
  Language Design and Implementation}, PLDI '04, pages 35--46, New York, NY,
  USA, 2004. ACM.
\newblock \href {http://dx.doi.org/10.1145/996841.996848}
  {\path{doi:10.1145/996841.996848}}.

\bibitem{10.1007/978-3-540-39403-7_19}
Miko{\l}aj Morzy, Tadeusz Morzy, Alexandros Nanopoulos, and Yannis
  Manolopoulos.
\newblock Hierarchical bitmap index: An efficient and scalable indexing
  technique for set-valued attributes.
\newblock In Leonid Kalinichenko, Rainer Manthey, Bernhard Thalheim, and Uwe
  Wloka, editors, {\em Advances in Databases and Information Systems}, pages
  236--252, Berlin, Heidelberg, 2003. Springer-Verlag.
\newblock \href {http://dx.doi.org/10.1007/978-3-540-39403-7_19}
  {\path{doi:10.1007/978-3-540-39403-7_19}}.

\bibitem{nagel_schr}
Kai Nagel and Michael Schreckenberg.
\newblock A cellular automaton model for freeway traffic.
\newblock {\em J. Phys. I France}, 2(12):2221--2229, Sept. 1992.
\newblock \href {http://dx.doi.org/10.1051/jp1:1992277}
  {\path{doi:10.1051/jp1:1992277}}.

\bibitem{master_th_patel}
Parag Patel.
\newblock Object oriented programming for scientific computing.
\newblock Master's thesis, The University of Edinburgh, 2006.

\bibitem{6339601}
Matt Pharr and William~R. Mark.
\newblock ispc: A {SPMD} compiler for high-performance {CPU} programming.
\newblock In {\em 2012 Innovative Parallel Computing (InPar)}, pages 1--13.
  IEEE Computer Society, May 2012.
\newblock \href {http://dx.doi.org/10.1109/InPar.2012.6339601}
  {\path{doi:10.1109/InPar.2012.6339601}}.

\bibitem{IJNC126}
Max Plauth, Frank Feinbube, Frank Schlegel, and Andreas Polze.
\newblock A performance evaluation of dynamic parallelism for fine-grained,
  irregular workloads.
\newblock {\em International Journal of Networking and Computing},
  6(2):212--229, 2016.
\newblock \href {http://dx.doi.org/10.15803/ijnc.6.2_212}
  {\path{doi:10.15803/ijnc.6.2_212}}.

\bibitem{Schafer:2013:RLD:2492045.2492052}
Henry Sch\"{a}fer, Benjamin Keinert, and Marc Stamminger.
\newblock Real-time local displacement using dynamic {GPU} memory management.
\newblock In {\em Proceedings of the 5th High-Performance Graphics Conference},
  HPG '13, pages 63--72, New York, NY, USA, 2013. ACM.
\newblock \href {http://dx.doi.org/10.1145/2492045.2492052}
  {\path{doi:10.1145/2492045.2492052}}.

\bibitem{Sengupta06awork-efficient}
Shubhabrata Sengupta, Aaron~E. Lefohn, and John~D. Owens.
\newblock A work-efficient step-efficient prefix sum algorithm.
\newblock In {\em Workshop on Edge Computing Using New Commodity
  Architectures}, 2006.

\bibitem{doi:10.1080/21580103.2016.1262793}
Hark-Soo Song and Sang-Hee Lee.
\newblock Effects of wind and tree density on forest fire patterns in a
  mixed-tree species forest.
\newblock {\em Forest Science and Technology}, 13(1):9--16, 2017.
\newblock \href {http://dx.doi.org/10.1080/21580103.2016.1262793}
  {\path{doi:10.1080/21580103.2016.1262793}}.

\bibitem{Spliet:2014:KDM:2588768.2576781}
Roy Spliet, Lee Howes, Benedict~R. Gaster, and Ana~Lucia Varbanescu.
\newblock {KMA}: A dynamic memory manager for {OpenCL}.
\newblock In {\em Proceedings of Workshop on General Purpose Processing Using
  GPUs}, GPGPU-7, pages 9:9--9:18, New York, NY, USA, 2014. ACM.
\newblock \href {http://dx.doi.org/10.1145/2576779.2576781}
  {\path{doi:10.1145/2576779.2576781}}.

\bibitem{Springer:2018:ICD:3178433.3178439}
Matthias Springer and Hidehiko Masuhara.
\newblock {Ikra-Cpp}: A {C++/CUDA DSL} for object-oriented programming with
  structure-of-arrays layout.
\newblock In {\em Proceedings of the 2018 4th Workshop on Programming Models
  for SIMD/Vector Processing}, WPMVP'18, pages 6:1--6:9, New York, NY, USA,
  2018. ACM.
\newblock \href {http://dx.doi.org/10.1145/3178433.3178439}
  {\path{doi:10.1145/3178433.3178439}}.

\bibitem{6339604}
Markus Steinberger, Michael Kenzel, Bernhard Kainz, and Dieter Schmalstieg.
\newblock {ScatterAlloc}: Massively parallel dynamic memory allocation for the
  {GPU}.
\newblock In {\em 2012 Innovative Parallel Computing (InPar)}, pages 1--10.
  IEEE Computer Society, May 2012.
\newblock \href {http://dx.doi.org/10.1109/InPar.2012.6339604}
  {\path{doi:10.1109/InPar.2012.6339604}}.

\bibitem{master_th_cuda_allc}
Radek Stibora.
\newblock Building of {SBVH} on graphical hardware.
\newblock Master's thesis, Faculty of Informatics, Masaryk University, 2016.

\bibitem{placement_delete}
Bjarne Stroustrup.
\newblock {Bjarne Stroustrup}'s {C++} style and technique {FAQ}. is there a
  ``placement delete"?
\newblock \url{http://www.stroustrup.com/bs_faq2.html\#placement-delete}, 2017.

\bibitem{STRZODKA2012429}
Robert Strzodka.
\newblock Chapter 31 - abstraction for {AoS} and {SoA} layout in {C++}.
\newblock In Wen mei W.~Hwu, editor, {\em GPU Computing Gems Jade Edition},
  Applications of GPU Computing Series, pages 429--441. Morgan Kaufmann,
  Boston, 2012.
\newblock \href {http://dx.doi.org/10.1016/B978-0-12-385963-1.00031-9}
  {\path{doi:10.1016/B978-0-12-385963-1.00031-9}}.

\bibitem{Tasos:2018:ESS:3242947.3242951}
Alexandros Tasos, Juliana Franco, Tobias Wrigstad, Sophia Drossopoulou, and
  Susan Eisenbach.
\newblock Extending {SHAPES} for {SIMD} architectures: An approach to native
  support for struct of arrays in languages.
\newblock In {\em Proceedings of the 13th Workshop on Implementation,
  Compilation, Optimization of Object-Oriented Languages, Programs and
  Systems}, ICOOOLPS '18, pages 23--29, New York, NY, USA, 2018. ACM.
\newblock \href {http://dx.doi.org/10.1145/3242947.3242951}
  {\path{doi:10.1145/3242947.3242951}}.

\bibitem{Ueno:2011:ENG:2034773.2034802}
Katsuhiro Ueno, Atsushi Ohori, and Toshiaki Otomo.
\newblock An efficient non-moving garbage collector for functional languages.
\newblock In {\em Proceedings of the 16th ACM SIGPLAN International Conference
  on Functional Programming}, ICFP '11, pages 196--208, New York, NY, USA,
  2011. ACM.
\newblock \href {http://dx.doi.org/10.1145/2034773.2034802}
  {\path{doi:10.1145/2034773.2034802}}.

\bibitem{Vinkler:2015:RED:3071494.3071506}
Marek Vinkler and Vlastimil Havran.
\newblock Register efficient dynamic memory allocator for {GPUs}.
\newblock {\em Comput. Graph. Forum}, 34(8):143--154, December 2015.
\newblock \href {http://dx.doi.org/10.1111/cgf.12666}
  {\path{doi:10.1111/cgf.12666}}.

\bibitem{Volkov:EECS-2016-143}
Vasily Volkov.
\newblock {\em Understanding Latency Hiding on GPUs}.
\newblock PhD thesis, EECS Department, University of California, Berkeley, Aug
  2016.
\newblock URL:
  \url{http://www2.eecs.berkeley.edu/Pubs/TechRpts/2016/EECS-2016-143.html}.

\bibitem{Weber:2017:MAL:3132652.3106341}
Nicolas Weber and Michael Goesele.
\newblock {MATOG}: Array layout auto-tuning for {CUDA}.
\newblock {\em ACM Trans. Archit. Code Optim.}, 14(3):28:1--28:26, August 2017.
\newblock \href {http://dx.doi.org/10.1145/3106341}
  {\path{doi:10.1145/3106341}}.

\bibitem{Widmer:2013:FDM:2458523.2458535}
Sven Widmer, Dominik Wodniok, Nicolas Weber, and Michael Goesele.
\newblock Fast dynamic memory allocator for massively parallel architectures.
\newblock In {\em Proceedings of the 6th Workshop on General Purpose Processor
  Using Graphics Processing Units}, GPGPU-6, pages 120--126, New York, NY, USA,
  2013. ACM.
\newblock \href {http://dx.doi.org/10.1145/2458523.2458535}
  {\path{doi:10.1145/2458523.2458535}}.

\bibitem{Zhu:2015:PIM:2817095.2817115}
Xiangyuan Zhu, Kenli Li, Ahmad Salah, Lin Shi, and Keqin Li.
\newblock Parallel implementation of {MAFFT} on {CUDA}-enabled graphics
  hardware.
\newblock {\em IEEE/ACM Trans. Comput. Biol. Bioinformatics}, 12(1):205--218,
  January 2015.
\newblock \href {http://dx.doi.org/10.1109/TCBB.2014.2351801}
  {\path{doi:10.1109/TCBB.2014.2351801}}.

\end{thebibliography}

\appendix

\section{Concurrency and Correctness}
\label{sec:concurrency}

CUDA has a weak consistency model for global memory access~\cite{Alglave:2015:GCW:2694344.2694391}. Writes to memory performed by one thread are \emph{not} guaranteed to become visible to other threads in the same order. However, atomic writes \emph{have} that property (\emph{sequential consistency}). Furthermore, \emph{thread fences} can be used between two memory writes to enforce sequential consistency, if necessary.

Moreover, global memory reads/writes may be buffered in registers/caches, without a global memory load/store. Thus, memory writes by one thread may not become visible to other threads until the next GPU kernel, unless reads/writes are \texttt{volatile} or performed with atomic operations.

All bitmap operations are sequentially consistent and do not suffer from load/store buffering because they are based on atomic memory operations.

\subsection{Object Slot Reservation/Freeing}
\label{sec:details_alloc_dealloc}
Inside a block, object allocations are tracked with the object allocation bitmap. Every object allocation bitmap has 64 bits, regardless of the block capacity. If a block's capacity is smaller than 64, then the last $64-N$ bits are set to 1 during block initialization to prevent threads from reserving these slots during object allocation.

Object slots are reserved/freed with atomic operations. These bypass the incoherent L1 caches and are thread-safe: E.g., based on their return value, we know if the current thread reserved a slot or if a contending thread was faster (Alg.~\ref{alg:block_allocate}, Line~5). Based on their return value, we also know if the current thread reserved the last slot (Line~11), in which case the block should be marked as inactive by the allocation algorithm.

\subsubsection{Slot Reservation}
\texttt{Block::reserve()} (Alg.~\ref{alg:block_allocate}) reserves a single object slot in the block. Our actual implementation may reserve multiple slots at once due to allocation request coalescing.

\begin{enumerate}
  \item \textbf{Preconditions:} Block was initialized at least once. (Calling this method on invalidated blocks or full blocks is OK. This function will simply return FAIL.)
  \item \textbf{Postconditions:} If the result is different from FAIL, the resulting slot at position is reserved for this thread (and no other thread).
  \item \textbf{Return Value:} Success indicator, atomically reserved slot position, block state.
  \item \textbf{Linearization Point:} Atomic OR operation (Line~5).
\end{enumerate}

\SetKwComment{Comment}{$\triangleright$\ }{}
\begin{algorithm}[t]
\small
\begin{multicols}{2}
    mask $\gets$ 1 {<}{<} pos\;
    before $\gets$ \emph{atomicAnd}(\&bitmap, $\sim$mask)\;
    \textcolor{gray}{success $\gets$ (before \& mask)) $\not=$ 0\;}
    \textcolor{gray}{\emph{assert}(success);} \hfill \Comment{\textcolor{gray}{\textsf{Precondition.\,\,\,\,\,\,\,\,}}}

    \uIf{popc(\emph{before}) \emph{= 1}}{
      \Return \emph{EMPTY}\;
    }
    \uElseIf{popc(\emph{before}) \emph{= 64}}{
      \Return \emph{FIRST}\;
    }
    \Else{
      \Return \emph{REGULAR}\;
    }
\end{multicols}
\vspace{0.25cm}
 \caption{Block::deallocate(pos) : state \hfill $\triangleright$\ \textsf{Assuming block size 64.} \,\,\,\fbox{GPU}}
 \label{alg:block_deleteslot}
\end{algorithm}

\subsubsection{Slot Freeing}
\texttt{Block::deallocate(pos)} (Alg.~\ref{alg:block_deleteslot}) frees a single object slot in the block. To support allocation request coalescing, we have a modified version of this function that can rollback multiple slots at once.

\begin{enumerate}
  \item \textbf{Preconditions:} Bit \texttt{pos} is set to 1 in the object allocation bitmap. (Deleting an object multiple times or trying to delete an arbitrary pointer is illegal.)
  \item \textbf{Postconditions:} Bit \texttt{pos} is set to 0 in the object allocation bitmap.
  \item \textbf{Return Value:} Block state.
  \item \textbf{Linearization Point:} Atomic AND operation (Line~2).
\end{enumerate}

\subsection{Safe Memory Reclamation with Block Invalidation}
\label{sec:block_invalidation_appx}
Safe memory reclamation (SMR) in lock-free algorithms is notoriously difficult. An SMR problem arises in \textsc{DynaSOAr} when deleting blocks. A block should be deleted as soon as its last object has been deleted. This by itself is easy to detect with atomic operations (Alg.~\ref{alg:block_deleteslot}, Line~6). However, a contending thread may already have selected the now empty block in the course of its own concurrent allocate operation, before the block is actually deleted. Now it is no longer safe to delete the block, but the deleting thread is not aware of that.

Elaborate techniques for SMR such as hazard pointers and epoch-based reclamation have been proposed in previous work~\cite{Brown:2015:RML:2767386.2767436, Michael:2002:SMR:571825.571829}. \textsc{DynaSOAr} is able to exploit a key characteristic of its data structure to solve this SMR problem in a simple way: Since all blocks have the same size and structure, object allocation bitmaps are always located at the same position. Therefore, we can optimistically proceed with bitmap modifications and rollback changes if necessary.

Our solution to SMR is \emph{block invalidation}. Before deleting a block, a thread tries to \emph{invalidate} (atomically set to 1) all bits in the object allocation bitmap. Bits that were already 1 are not considered invalidated because those object slots are in use. After successful invaldation, bits remain invalidated until a new block is initialized in the same location. Other threads may still be able to find the block in the active block bitmap for a while, but object slot reservations can no longer succeed. 

\SetKwComment{Comment}{$\triangleright$\ }{}
\begin{algorithm}[t]
\small
  heap[bid].type $\gets$ T; \hfill \Comment{\textsf{Volatile write.}}
  \emph{\_\_threadfence}()\;
  heap[bid].bitmap $\gets$ 0; \hfill \Comment{\textsf{Volatile write, assuming block capacity 64.}}
 \caption{DAllocatorHandle::initialize\_block<T>(int bid) : void \hfill \fbox{GPU}}
 \label{alg:block_init}
\end{algorithm}

\SetKwComment{Comment}{$\triangleright$\ }{}
\begin{algorithm}[t]
\small
  bitmap\_ptr $\gets$ \&heap[bid].bitmap\;
  before $\gets$ \emph{atomicOr}(bitmap\_ptr, 0xFF...F); \hfill \Comment{\textsf{Invalidate (set) all obj. allocation bitmap bits.}}
  \If(\Comment*[f]{\textsf{$\geq$ 1 bit was invalidated.}}){before $\not=$ 0xFF...F}{
    t $\gets$ heap[bid].type\;
    \uIf(\Comment*[f]{\textsf{All 64 bits invalidated by this \emph{atomicOr}.}}){\emph{before = 0}}{
      \Return \emph{true}\;
    }
    \Else(\Comment*[f]{\textsf{Not all bits invalidated. Rollback.}}){
      before\_rollback $\gets$ \emph{atomicAnd}(bitmap\_ptr, before)\;
      \If(\Comment*[f]{\textsf{Other thread cleared a bit.}}){\emph{before\_rollback $\not=$ 0xFF...F}}{
            active[t].\emph{clear}(bid); \hfill \Comment{\textsf{Other thread expects an inactive block.}}
          }
      \If(\Comment*[f]{\textsf{Empty again. Retry invalidation.}}){\emph{(before\_rollback \& before) = 0}}{
        \Return \emph{invalidate}(bid)\;
      }
    }
  }
\Return \emph{false}\;
 \caption{DAllocatorHandle::invalidate(int bid) : bool \hfill \fbox{GPU}}
 \label{alg:block_invalidate}
\end{algorithm}

Allocating threads can detect changes in the block type. Before a previously invalidated block becomes available for allocations again (by initializing its object allocation bitmap), we update the block type. We put a thread fence between both writes to ensure that threads see the new block type before they see free slots in the bitmap (Alg.~\ref{alg:block_init}). Threads allocate objects optimistically and rollback changes should they detect a different block type (Alg.~\ref{alg:alloc_algo}, Line~14; also see Sec.~\ref{sec:cor_obj_alloc}).

\paragraph*{Details}
Block invalidation\footnote{For presentation reasons, we assume a block capacity of 64 in all algorithms in this paper.} (Alg.~\ref{alg:block_invalidate}) fails if a thread is unable to invalidate at least one bit. In that case, if at least one bit was changed through invalidation, this change must be rolled back (Line~8): In \texttt{before} exactly those bits are zero that were invalidated by the thread.

While a thread is running an invalidation operation, other threads may continue to concurrently reserve/free object slots in the same block, unaware of the fact that a thread is trying to invalidate the block. Those threads will update block bitmaps based on the object allocation bitmap state that they are seeing. Therefore, block invalidation must update block bitmaps, as every invalidated bit appears to be an allocated object slot to other threads.

Since block invalidation fills up a block, the block's \emph{active[t]} state should be removed after Line~7, because, if we enter this \emph{else} branch, the thread just \emph{filled up} the block by reserving the remaining object slots (however, not all 64~slots, otherwise, we would be in  the \emph{then} branch of Line~5). However, we defer this step, as an invalidation rollback would likely have to mark the same block as \emph{active[t]} again. Unless, another thread concurrently freed an object slot in-between invalidation and invalidation rollback. For such a thread it will seem as if its deallocation just freed the first slot, causing it activate the block (Alg.~\ref{algo:dealloc_a}, Line~5). However, since we defered block deactivation, this \emph{set(bid)} operation will spin until we deactivate the block (Alg.~\ref{alg:block_invalidate}, Line~10). If invalidation rollback empties the block again, we try to invalidate the block one more time\footnote{Our actual implementation is iterative instead of recursive.}.

Note that block invalidation is independent of the type of a block. After invalidating at least one bit, the block type is fixed until invalidation rollback or block initialization, since other threads do not change invalidated bits. As such, the block cannot be deleted or reinitialized to another type by another thread. Other threads can, however, delete and initialize a block with different type after invalidation rollback. It is, nevertheless, safe to assume a block type of \emph{t} in Line~10, since this is merely an execution of a defered operation that should have happened earlier when the block type was known to be \emph{t}.



\subsection{Object Allocation}
\label{sec:cor_obj_alloc}
The critical parts during allocations (Alg.~\ref{alg:alloc_algo}) are \emph{block selection} (Line~2) and \emph{object slot reservation} (Line~8). Both operations by themselves are atomic, but not together. Block selection returns the index of an active block of type $T$, so we expect that after Line~8, we reserved an object slot in a block of type $T$. However, due to concurrent operations of other threads, some of these assumptions may be violated.

\begin{description}
  \item[Block Full] An active block was selected by \texttt{try\_find\_set} but the block filled up before making an allocation (i.e., the block is no longer active). In this case, object slot reservation will fail. Whenever allocation fails, it will restart from the beginning.
  \item[Block Deallocated] A block was selected by \texttt{try\_find\_set} but deallocated before reserving a slot. In this case, slot reservation will fail because the block is now in an invalidated state.
  \item[Block Replaced (ABA)] A block was selected by \texttt{try\_find\_set} but deallocated and reinitialized to a block of same type $T$. This is harmless: We do not care about block identity.
  \item[Block Replaced (Different Type)] A block was selected by \texttt{try\_find\_set} but deallocated and reinitialized to another type\footnote{Block initialization (Alg.~\ref{alg:block_init}) has a thread fence between setting the block type and resetting the object allocation bitmap, so threads are guaranteed to read the correct type $t$ after an allocation succeeded.} $t \not= T$. In this case, the allocation must be rolled back (Line~14). All blocks have the same basic structure, so no data can be overwritten accidentally during bitmap updates. Note that the rollback may trigger additional block bitmap updates.
  \textcolor{gray}{\item[Active Block Not Selected] A block becomes active shortly after \texttt{try\_find\_set} fails. Or, due to bitmap hierarchy inconsistencies, \texttt{try\_find\_set} fails to find an active block even though active blocks exist. This is harmless: No assumption is violated. A new block will be initialized, which merely increases fragmentation.}
\end{description}

Note that a block cannot be deallocated after an object slot was already reserved, because block invalidation would fail. Thus, the type of a block can also no longer change.


\subsection{Object Deallocation}
The critical part during deallocations (Alg.~\ref{algo:dealloc_a}) is consistency between \emph{object slot deallocation} (Line~3) and \emph{block state updates}. If the current thread deallocated the first object (i.e., the block was full), then the block bit must be set to active. If the current thread deleted the last object (i.e., the block is empty), then the block must be deleted. The problem is that object slot deallocation and the corresponding block state update together are not atomic.

\begin{description}
  \item[Allocate After Delete-First] A thread $t_1$ deleted the first object of a block. However, before marking the block active (Line~6), another thread $t_2$ allocated this slot again; the block should be inactive. In this case, $t_2$ reserved the last slot, so it will mark the block as inactive (Alg.~\ref{alg:alloc_algo}, Line~12). This operation expects the bit to be in a set state and it will retry until $t_1$ sets the bit.
  \item[Block Deleted after Delete-First] A thread $t_1$ deleted the first object of a block. However, before marking the block active, other threads deallocated all other objects and a thread $t_2$ deleted the block. This is not possible because $t_2$ expects the block to be active (Line~9), i.e., bit set to 1, and blocks until then.
  \item[Block Replaced after Delete-First] A thread $t_1$ deleted the first object of a block. However, before marking the block active, the block was reinitialized to another type. This is not possible because only deleted blocks can be reinitialized (see previous point).
  \item[Allocate after Delete-Last] A thread $t_1$ deleted the last object of a block. However, before deleting the block, another thread $t_2$ allocated an object again, so it is unsafe to delete the block now. This case in handled by block invalidation.
  \item[Block Deleted after Delete-Last] A thread $t_1$ deleted the last object of a block. However, before deleting the block, another thread $t_2$ allocated an object and yet another thread $t_3$ deleted that object, rendering the block empty again and deleting it. Now the block is already deleted when $t_1$ is trying to delete the block. In this case, block invalidation of $t_1$ will fail because the block is still in an invalidated state and $t_1$ fails to invalidate all object slot bits.
  \item[Block Replaced after Delete-Last] Same as before, but yet another thread $t_4$ reininitializes the block to a different type. Now $t_1$ will invalidate and delete a new block whose type is different. This is OK. Block invalidation will succeed only if the block is empty. Both block invalidation and block deletion are independent of the block type.
\end{description}

\subsection{Correctness of Hierarchical Bitmap Operations}
A container $C_i^l$ consists of bits $b_{64 \cdot i}^l$, ...,  $b_{64 \cdot i + 63}^l$ and is represented by one bit $b_{i}^{l+1}$ in the nested (higher-level) bitmap. That bit is set if and only if at least one bit is set in the container.

\begin{definition}[Consistency]
\label{def:consistency_crit}
A bit in level $b_{i}^{l+1}$ is \textbf{consistent} with its corresponding container $C_i^l$ in the lower-level bitmap if and only if:

\begin{align*}
b_i^{l+1} = \bigvee_{k=0}^{63} b_{64\cdot i + k}^{l}\,\, \textcolor{gray}{= \mathds{1}\left(\sum C_{\lfloor i/64 \rfloor}^l > 0\right)}
\end{align*}
\end{definition}

We say that the $L_{l+1}$ bitmap is in a consistent state with the $L_l$ bitmap if all bits $b_i^{l+1}$ in the $L_{l+1}$ bitmap satisfy the consistency criterion. The bitmap data structure as a whole is in a consistent state if all bitmap levels $L_i$ satisfy the consistency criterion.

\begin{definition}[Semantics of Bitmap Operations]
\label{def:semant_bitmap_ops}
Every bitmap $L_l$ provides operations for setting and clearing bits (Sec.~\ref{sec:bitmap_operations_412}). These operations may update bits in the higher-level bitmap $L_{l+1}$ if they \textbf{s}et the \textbf{f}irst bit ($\mathit{SF}_{\lfloor i/64 \rfloor}^l$) or \textbf{c}lear the \textbf{l}ast bit ($\mathit{CL}_{\lfloor i/64 \rfloor}^l$) of a container $C^l_i$, respectively:

\begin{align*}
\underbrace{\textsf{set}(b^l_i) \text{ and } \mathds{1}\left(\sum C_{\lfloor i/64 \rfloor}^l = 0\right)}_{\text{set-first: } \mathit{SF}_{\lfloor i/64 \rfloor}^l} & \text{ then } \textsf{set}(b_{\lfloor i/64 \rfloor}^{l+1}) & \,\,\,\,\,\, \forall i \in [0; 64) \\
\underbrace{\textsf{clear}(b^l_i) \text{ and } \mathds{1}\left(\sum C_{\lfloor i/64 \rfloor}^l = 1\right)}_{\text{clear-last: } \mathit{CL}_{\lfloor i/64 \rfloor}^l} & \text{ then } \textsf{clear}(b_{\lfloor i/64 \rfloor}^{l+1}) & \,\,\,\,\,\, \forall i \in [0; 64)
\end{align*}
\end{definition}

We would like to show that, assuming that a bitmap data structure is initially in a consistent state and given a multiset of bitmap operations $O_0$ on the $L_0$ bitmap, the entire bitmap data structure is in a consistent state after executing all operations.

\begin{definition}[Legal Bitmap Operations]
\label{def:correctness_multiset_op}
Let $\#\textsf{set}(b^l_i)$ and $\#\textsf{clear}(b^l_i)$ be the number of set and clear operations of $b^l_i$ in a multiset of bitmap operations $O_l$. We call $\mathds{S}(b^l_i) = \#\textsf{set}(b^l_i) - \#\textsf{clear}(b^l_i)$ the \textbf{set-surplus} of $b^l_i$. $O_l$ is \textbf{legal} if it satifies the following conditions.

\begin{enumerate}
  \item Overall bit operation is clear, remain or set: $\mathds{S}(b^l_i) \in \{-1, 0, 1\}$.
  \item Bit is in a cleared or set state afterwards: $b^l_i + \mathds{S}(b^l_i) \in \{0, 1\}$.
\end{enumerate}
\end{definition}

E.g., setting a cleared bit twice and clearing it once ($\mathds{S} = 2 - 1 = 1$ and $0 + 1 = 1$) is OK, but setting the bit three times and clearing it once ($\mathds{S} = 3 - 1 = 2$) would be an illegal usage of the bitmap data structure. Note that illegal bitmap operations deadlock in our implementation because \textsf{set} and \textsf{clear} spin-block and retry until they acutally changed the bit. If a legal bitmap operations multiset is executed fully concurrent (i.e., one thread per operation), then there is always a thread/operation that can make progress.

\begin{hypothesis}
Let us assume that a multiset of bitmap operations $O_l$ on the $L_l$ bitmap is legal according to Definition~\ref{def:correctness_multiset_op} for an arbitrary $l$ and that $L_l$ is initially consistent with $L_{l+1}$.
\end{hypothesis}

\begin{lemma}
Under the induction hypothesis, the bitmap operations multiset $O_{l+1}$ that is generated by the operations in $O_l$ according to Definition~\ref{def:semant_bitmap_ops} is also legal. Furthermore, after executing $O_l$, $L_l$ is still consistent with $L_{l+1}$.
\end{lemma}

\begin{proof}
Let us first consider the bitmap operations of a single container $C^l_i$. Let $\#\mathit{SF}_i^l$ be the number of times a first bit is set in the container and $\#\mathit{CL}_i^l$ be the number of times a last bit is cleared in the container. Then, according to Definition~\ref{def:semant_bitmap_ops}, $b_{\lfloor i/64 \rfloor}^{l+1}$ is set $\#\mathit{SF}_i^l$ times and cleared $\#\mathit{CL}_i^l$ times. We have to prove that the set-surplus $\mathds{S}(b^{l+1}_{\lfloor i/64 \rfloor}) = \#\mathit{SF}_i^l - \#\mathit{CL}_i^l$ satisfies the legality criteria of Definition~\ref{def:correctness_multiset_op}.

Without loss of generality, let us assume that all set-first and clear-last operate on the same bit $b^l_k$. Then, $\mathds{S}(b^{l+1}_{\lfloor i/64 \rfloor}) = \mathds{S}(b^l_k) \in \{-1, 0, 1\}$. Hence, the generated bitmap operations $O_{l+1}$ for any bit on the $L_{l+1}$ bitmap satisfy the first legality condition of Definition~\ref{def:correctness_multiset_op}.

Now we have to show that also the second legality condition holds and that $b^{l+1}_{\lfloor i/64 \rfloor}$ is consistent with $C^l_i$ after executing $O_l$. We consider two cases.

\begin{enumerate}
  \item $b^{l+1}_{\lfloor i/64 \rfloor} = 0$. Therefore, due to initial consistency, $\sum C_{\lfloor i/64 \rfloor}^l = 0$. Therefore, $\#\mathit{SF}_i^l - \#\mathit{CL}_i^l \in \{0, 1\}$, otherwise, $O_l$ would not be legal. Therefore, $b^{l+1}_{\lfloor i/64 \rfloor} + \mathds{S}(b^{l+1}_{\lfloor i/64 \rfloor}) \in \{0, 1\}$.
  \begin{enumerate}
    \item If $\#\mathit{SF}_i^l - \#\mathit{CL}_i^l = 0$, then $\vee_{k=0}^{63} b_{64\cdot i + k}^{l} = 0$ after $O_l$. At the same time, $\mathds{S}(b^{l+1}_{\lfloor i/64 \rfloor}) = 0$, so $b^{l+1}_{\lfloor i/64 \rfloor} = 0$ after $O_l$, which is consistent with the state of $C^l_i$ after $O_l$.
    \item If $\#\mathit{SF}_i^l - \#\mathit{CL}_i^l = 1$, then $\vee_{k=0}^{63} b_{64\cdot i + k}^{l} = 1$ after $O_l$. At the same time, $\mathds{S}(b^{l+1}_{\lfloor i/64 \rfloor}) = 1$, so $b^{l+1}_{\lfloor i/64 \rfloor} = 1$ after $O_l$, which is consistent with the state of $C^l_i$ after $O_l$.
  \end{enumerate}
  \item $b^{l+1}_{\lfloor i/64 \rfloor} = 1$. Therefore, due to initial consistency, $\sum C_{\lfloor i/64 \rfloor}^l > 0$. Therefore, $\#\mathit{SF}_i^l - \#\mathit{CL}_i^l \in \{-1, 0\}$, otherwise, $O_l$ would not be legal. Therefore, $b^{l+1}_{\lfloor i/64 \rfloor} + \mathds{S}(b^{l+1}_{\lfloor i/64 \rfloor}) \in \{0, 1\}$.
  \begin{enumerate}
    \item If $\#\mathit{SF}_i^l - \#\mathit{CL}_i^l = -1$, then $\vee_{k=0}^{63} b_{64\cdot i + k}^{l} = 0$ after $O_l$. At the same time, $\mathds{S}(b^{l+1}_{\lfloor i/64 \rfloor}) = -1$, so $b^{l+1}_{\lfloor i/64 \rfloor} = 0$ after $O_l$, which is consistent with the state of $C^l_i$ after $O_l$.
    \item If $\#\mathit{SF}_i^l - \#\mathit{CL}_i^l = 0$, then $\vee_{k=0}^{63} b_{64\cdot i + k}^{l} = 1$ after $O_l$. At the same time, $\mathds{S}(b^{l+1}_{\lfloor i/64 \rfloor}) = 0$, so $b^{l+1}_{\lfloor i/64 \rfloor} = 1$ after $O_l$, which is consistent with the state of $C^l_i$ after $O_l$.
  \end{enumerate}
\end{enumerate}

If all containers in $L_l$ are consistent with their respective bits in $L_{l+1}$, then the entire $L_l$ bitmap is consistent with the $L_{l+1}$ bitmap. Futhermore, all generated bitmap operations $O_{l+1}$ are legal because they satisfy both legality criteria.
\end{proof}

\begin{basecase}
The bitmap data structure is initially in a consistent state. Furthermore, $O_0$ is legal. Otherwise, programmers use the bitmap data structure incorrectly.
\end{basecase}

\section{Field Address Computation}
\label{sec:address_computation_ex}
This section describes a key implementation technique of the \textsc{DynaSOAr} DSL, that was taken from Ikra-Cpp~\cite{Springer:2018:ICD:3178433.3178439}: Proxy Types. This technique allows us to implement custom data layouts in C++~11 without breaking OOP abstractions or modifying the compiler.

Even though fields are declared with type \texttt{Field<B, N>}, they can be used almost like normal C++ types. There are certain limitations with respect to automatic type deduction (\texttt{auto} keyword). Internally, this is implemented with operator overloading, e.g.:

\begin{enumerate}
  \item \textbf{Implicit Conversion Operator:} \texttt{Field<B, N>} values can be implicitly converted to the N-th predeclared type in \texttt{B}, without an explicit type cast. We call \texttt{B} the \emph{base type}.
  \item \textbf{Member of Object/Pointer Operators:} It is possible to call non-virtual member functions if the base type is (pointer to) a class or struct.
  \item \textbf{Subscript Operator:} It is possible to use array access syntax (\texttt{[]}) for array base types.
  \item \textbf{Indirection/Address-of Operators:} It is possible to dereference a value of pointer base type and to take the address of a field value.
\end{enumerate} 

Listing~\ref{lst:addr_comp_ft} shows the implementation of the implicit conversion operator. This code first extracts all components that are required for address computation from an object pointer. Then it returns a reference to an object of the base type at the computed memory location.
\bigskip
\begin{lstlisting}[caption={Address Computation in Proxy Field Types}, label={lst:addr_comp_ft}]
// Implicit conversion operator: E.g., convert Field<NodeBase, 2> to float& in Figure <@\ref{fig:global_ref}@>.
template<typename B, int N>
Field<B, N>::operator typename B::predeclared_type<N>&() {
  int offset = ...;  // Computed with template metaprogramming. <@$\textsf{offset}_{\textsf{B::fieldname}}$@> in Figure <@\ref{fig:global_ref}@>.
  auto obj_ptr = reinterpret_cast<uint64_t>(this) - 2;  // p2 in Figure <@\ref{fig:global_ref}@>.
  // Bits 0-49 and clear 6 least significant bits.
  auto* block_address = reinterpret_cast<char*>(obj_ptr & 0x3FFFFFFFFFFC0);
  int obj_slot_id = obj_ptr & 0x3F;  // Bits 0-5
  int block_capacity = (obj_ptr & 0xFC000000000000) >> 50;  // Bits 50-55
  auto* soa_array = reinterpret_cast<typename B::predeclared_type<N>*>(
      block_address + field_offset * block_capacity);
  return soa_array[obj_slot_id];
}
\end{lstlisting}

\end{document}